\pdfoutput=1
\documentclass{amsart}

\usepackage{lmodern}
\usepackage{amscd}
\usepackage{bbm}
\usepackage{mathrsfs}
\usepackage{eucal}
\usepackage{faktor}
\usepackage{xfrac}
\usepackage{braket}
\usepackage{relsize} 
\usepackage{nccmath}
\usepackage{mathtools}
\usepackage{comment}
\usepackage{verbatim}
\usepackage{wasysym}
\usepackage{graphicx}
\usepackage{caption}
\usepackage{import}
\usepackage{latexsym}
\usepackage{amsfonts}
\usepackage{amssymb}
\usepackage{amsmath}
\usepackage{mathtext}
\usepackage{cite}
\usepackage{enumerate}
\usepackage{float}
\usepackage{amsthm}
\usepackage{upgreek}
\usepackage{nicefrac}
\usepackage{multicol}
\usepackage{rotating}
\usepackage{array}
\usepackage{microtype}
\usepackage{vwcol} 
\usepackage[normalem]{ulem}

\usepackage{etoolbox}
\makeatletter
\pretocmd{\chapter}{\addtocontents{toc}{\protect\addvspace{15\p@}}}{}{}
\pretocmd{\section}{\addtocontents{toc}{\protect\addvspace{5\p@}}}{}{}
\makeatother
\let\oldtocsection=\tocsection
\let\oldtocsubsection=\tocsubsection
\let\oldtocsubsubsection=\tocsubsubsection
\renewcommand{\tocsection}[2]{\hspace{0em}\oldtocsection{#1}{#2}}
\renewcommand{\tocsubsection}[2]{\hspace{1.8em}\oldtocsubsection{#1}{#2}}
\renewcommand{\tocsubsubsection}[2]{\hspace{4.4em}\oldtocsubsubsection{#1}{#2}}

\usepackage[svgnames]{xcolor}
\definecolor{linkcolor}{HTML}{e88d67} 
\definecolor{citecolor}{HTML}{e88d67} 
\definecolor{urlcolor}{HTML}{e88d67} 
\definecolor{myNewColorA}{HTML}{fec3a6}
\definecolor{myNewColorB}{HTML}{ffaf80}
\definecolor{myNewColorC}{HTML}{fb8f67}
\definecolor{seagreen}{HTML}{337180}
\definecolor{mseagreen}{HTML}{369673}
\definecolor{darksalmon}{HTML}{e88d67}
\definecolor{silver}{HTML}{bbbbbb}
\definecolor{flowerblue}{HTML}{4e77fc}
\definecolor{tomato}{HTML}{ff6347}
\definecolor{orange}{HTML}{f2b13d}
\definecolor{darkgray}{HTML}{939393}

\usepackage[
  bookmarks = true,
  bookmarksopen = true,
  colorlinks = , 
  urlcolor = urlcolor, 
  linkcolor = linkcolor, 
  citecolor = citecolor, 
  linktoc = all,
  pagebackref = true
]{hyperref} 
\mathtoolsset{showonlyrefs,showmanualtags}

\DeclareMathOperator{\Mat}{Mat}

\DeclareMathOperator{\const}{const}

\newcommand{\br}[1]{\left( #1 \right)}

\newcommand{\lbr}[1]{\left[ #1 \right]}

\newcommand{\cLen}[1]{\mathcal{L}_{ #1 }}
\newcommand{\ncLen}[1]{\mathcal{L}_{ #1 }^{\text{\tiny{nc}}}}
\newcommand{\ncLenard}[1]{\mathcal{L}_{ #1 }^{\text{\tiny{nc}}} \lbr{u}}

\newcommand{\mseagreen}[1]{{\color{mseagreen} #1}}

\newcommand{\flowerblue}[1]{{\color{flowerblue} #1}}
\newcommand{\tomato}[1]{{\color{tomato} #1}}
\newcommand{\orange}[1]{{\color{orange} #1}}

\theoremstyle{plain}
\newtheorem{thm}{Theorem}[]

\newtheorem{lem}{Lemma}[]

\newtheorem{prop}{Proposition}[]

\theoremstyle{definition}
\newtheorem{defn}{Definition}[]

\newtheorem{exmp}{Example}[]

\theoremstyle{remark}
\newtheorem{rem}{Remark}

\usepackage{chngcntr}
\counterwithin*{thm}{section} 
\counterwithin*{lem}{section} 
\counterwithin*{claim}{section} 
\counterwithin*{prop}{section} 
\counterwithin*{corl}{section} 
\counterwithin*{defn}{section} 
\counterwithin*{exmp}{section} 
\counterwithin*{ex}{section} 
\counterwithin*{rem}{section} 

\usepackage{apptools}
\AtAppendix{\counterwithin{thm}{section}}
\AtAppendix{\counterwithin{lem}{section}}
\AtAppendix{\counterwithin{claim}{section}}
\AtAppendix{\counterwithin{prop}{section}}
\AtAppendix{\counterwithin{corl}{section}}
\AtAppendix{\counterwithin{defn}{section}}
\AtAppendix{\counterwithin{exmp}{section}}
\AtAppendix{\counterwithin{ex}{section}}
\AtAppendix{\counterwithin{rem}{section}}

\allowdisplaybreaks

\usepackage{tikz}

\begin{document}

    \title[A non-commutative discrete first Painlevé hierarchy]{A non-commutative discrete first Painlevé hierarchy: \\
    the Lax pair approach}

    \author{Irina Bobrova}
    \noindent\address{\noindent 
    Technische Universität Berlin, Institut für Mathematik, MA 7-1, Str. des 17. Juni 136, 10623 Berlin, Germany
    }
    \email{irrrina.bobrova@gmail.com}

    
    \begin{abstract}
    Using a non-commutative analogue of the isomonodromic problem associated with the discrete first Painlevé hierarchy, we construct a non-commutative version of this hierarchy, denoted by \ref{eq:dPIm_nc}. We show that both hierarchies, \ref{eq:dPIm} and \ref{eq:dPIm_nc}, can be expressed in terms of the polynomials $S_s^k(n)$, which we call the Svinin polynomials. We also derive a reduction of the non-commutative Volterra lattice hierarchy to the \ref{eq:dPIm_nc} hierarchy and present explicit continuous limits for the first three members of the \ref{eq:dPIm_nc}, thereby recovering non-commutative analogues of the first three members of the differential first Painlevé hierarchy.
    \end{abstract}
    
    \maketitle
    \pagestyle{headings}
    
    \tableofcontents
    
    \section{Introduction}

The celebrated differential and discrete Painlevé equations are among the central objects studied in modern mathematics and mathematical physics. Their importance stems from their rich structures---they possess symmetries generating affine Weyl groups, are related to special functions and orthogonal polynomials, can be interpreted as moduli spaces of flat connections, character varieties and Poisson geometry on which are one of the interests in modern mathematics---as well as their deep connections with integrable systems such as the Toda chains, the Kadomtsev–Petviashvili equation, Calogero systems, and many others. Thanks to the works of K.~Okamoto and H.~Sakai, these equations also admit a beautiful geometric interpretation. Moreover, the Painlevé equations possess natural higher-order extensions known as the \emph{Painlevé hierarchies}. We refer the reader to the excellent survey on this topic given in \cite{conte2008painleve}.

\medskip
In recent years, matrix and non-commutative generalizations of integrable systems have attracted growing interest. This is motivated both by applications and by the intrinsic mathematical richness of evolution equations with non-commuting dependent variables. The Painlevé equations, together with their discrete versions, provide a particularly compelling example of this phenomenon. For instance, the first Painlevé equation admits the following matrix analogue~\cite{Balandin_Sokolov_1998}:
\begin{align}
    \label{eq:PI1_mat}
    \tag*{$\text{PI}_1^{\text{mat}}$}
    &&
    u''
    &= 6 u^2 + z \, \mathbb{I} + a,
    &
    u=u(z), \,\, a 
    &\in \Mat_k(\mathbb{C}),
    &
    z
    &\in\mathbb{C}
    .
    &&
\end{align}
When the size of the matrices is equal to 1, it recovers the famous Painlevé I equation. The commutative first Painlevé equation has higher analogues discovered in \cite{kudryashov1997first} and later extended to the matrix case in \cite{Pic}. Both papers are based on the idea of making a symmetry reduction of the corresponding, scalar or matrix, Korteweg-de Vries hierarchy. The matrix hierarchy contains the \ref{eq:PI1_mat} equation as the first member and is called the matrix first Painlevé hierarchy. 

\medskip
The~\ref{eq:PI1_mat} equation can be further extended to a non-commutative setting by working in a unitary free associative algebra with a derivation satisfying the Leibniz rule (see, e.g.,~\cite{bobrova2022classification}). It also admits several discrete versions. One of the discrete matrix analogues is
\begin{align}
    \label{eq:dPI10_mat}
    &&
    \br{
    \alpha_{n + 1} + \alpha_n + \alpha_{n - 1}
    } \alpha_n
    + t \, \alpha_n
    &= c_n \, \mathbb{I}
    ,
    &
    c_n
    &= c_0 + n,
    &
    \alpha_n
    &\in \Mat_k(\mathbb{C})
    ,
    &
    t, \,\, c_0
    &\in \mathbb{C}
    ,
    &
    n
    &\in \mathbb{Z}
    .
    &&
\end{align}
If $k=1$, it turns into the discrete first Painlevé equation, the \ref{eq:dPI1}, which is the special case of the d-P$(E_6^{(1)})$ equation from the H. Sakai classification \cite{sakai2001rational}. 
The equation \eqref{eq:dPI10_mat} exhibits the singularity confinement property~\cite{cassatella2014singularity}, is related to matrix orthogonal polynomials~\cite{manas2021matrix}, and can be obtained from a suitable composition of Bäcklund transformations for the non-commutative fourth Painlevé system~\cite{bobrova2024affine}. 

Two additional matrix analogues of the discrete first Painlevé equation were obtained via symmetry reductions of the Volterra lattices~\cite{adler2020}. They are given by
\begin{align}
    \label{eq:dPI11_m}
    &&
    \alpha_{n + 1}\alpha_n 
    + \alpha_n^2 
    + \alpha_n\alpha_{n - 1}
    + t \, \alpha_n
    &= \tilde c_n \, \mathbb{I}
    ,
    &
    \tilde c_n
    &= c_0 + n + c_1 (-1)^n
    ,
    &&
    \\
    \label{eq:dPI12_m}
    &&
    \alpha_{n + 1}^\text{T}\alpha_n 
    + \alpha_n^2
    + \alpha_n\alpha_{n - 1}^\text{T}
    + t \, \alpha_n
    &= \tilde c_n \, \mathbb{I}
    ,
    &
    \tilde c_n
    &= c_0 + n + c_1 (-1)^n
    ,
    &&
\end{align}
where $c_0$ and $c_1$ are constant parameters. 
Because the matrix Volterra lattices are integrable—admitting higher symmetries and a zero-curvature representation—and the reduction extends for the Lax pairs as well, both equations inherit isomonodromic Lax pairs (see Proposition~3 in~\cite{adler2020}).

\medskip
In this work, we construct a non-commutative analogue of the discrete first Painlevé hierarchy, following an approach based on a discrete isomonodromic problem, which was introduced and applied to the discrete first and second Painlevé equations in~\cite{joshi1992nonlinear} and later extended to their higher-order versions in~\cite{cresswell1999discrete}. Remarkably, the first non-trivial member \ref{eq:dPI1_nc} of our hierarchy \ref{eq:dPIm_nc} coincides with a non-commutative extension of equation~\eqref{eq:dPI11_m}. Below, we provide a brief overview of the results and the structure of the paper.

\medskip 
We set up the non-commutative framework as follows. Let $\mathcal{R}$ be an associative unital division ring over a field $F$ with $\mathrm{char}(F)=0$. This may be viewed as a non-commutative generalization of the field of rational functions. We identify the unit of $F$ with that of $\mathcal{R}$, and note that the elements of $\mathcal{R}$ need not commute. We refer to elements of $\mathcal{R}$ informally as ``functions''.

To work with non-commutative analogues of ordinary differential and discrete equations, we introduce two operations on $\mathcal{R}$: a derivation $\partial$ and a translation $T$. A \textit{derivation} of $\mathcal{R}$ is an $F$-linear map $\partial : \mathcal{R} \to \mathcal{R}$ satisfying the Leibniz rule, and for any $c \in F$ we have $\partial (c) = 0$. Sometimes we use $\partial \equiv \prime$ if the subindex does not need to be specified additionally. 
Let the function $\alpha_n \in \mathcal{R}$ be an element at site $n \in \mathbb{Z}$, in other words it is a non-commutative function depending on the index $n$. On the lattice $\mathbb{Z}$ we consider a \textit{translation operator}~$T$, which is an $F$-linear map $T: \mathcal{R} \to \mathcal{R}$ such that $T\br{\alpha_n} = \alpha_{n + 1}$ and $T\br{c} = c$ for any constant $c$. For the simplicity, we take $F = \mathbb{C}$. 

\medskip
Following~\cite{joshi1992nonlinear}, we begin with the linear system
\begin{align}
    \label{eq:lsys_nc}
    \left\{
    \begin{array}{rcl}
    \alpha_n w_{n + 1}
    &=& \lambda \, w_n 
    - w_{n - 1},
    \\[2mm]
    \partial_{\lambda} w_n
    &=& a_n \, w_{n + 1} 
    + b_n \, w_n
    .
    \end{array}
    \right.
\end{align}
Here $\lambda \in \mathbb{C}^\times$ is a spectral parameter, and the functions $\alpha_n$, $a_n$, $b_n$, $w_n$ belong to the division ring $\mathcal{R}$. In addition, only $\partial_\lambda w_n \neq 0$. The compatibility condition of the system \eqref{eq:lsys_nc} gives rise to a system of equations for the functions $\alpha_n$, $a_n$, $b_n$. In order to solve them, we present the functions as polynomial series in $\lambda$ of highest degree $(2m+1)$ with $m \in \mathbb{N}$. Since $\lambda$ is an arbitrary non-zero parameter, the coefficients of each power of $\lambda$ must vanish. This yields a system of recurrence equations, which can be solved explicitly. Once $a_n$, $b_n$ are determined, the compatibility condition gives rise to a single equation for the function $\alpha_n$, which we refer to as a \textit{non-commutative analogue of the $m$-th member of the discrete first Painlevé hierarchy}. The hierarchy can be expressed in terms of a non-commutative version of the \textit{Svinin polynomials}~$S_s^k(n)$ whose commutative counterparts were introduced in \cite{svinin2009some}. These polynomials are defined as follows (see~Definition~\ref{def:spoly_nc}; see~also~\cite{casati2020recursion}):
\begin{align}
    \label{eq:spoly}
    &&
    S_{s}^0 \br{n}
    &= 1,
    &
    S_{s}^{k} \br{n}
    &= \sum_{j = 0}^{s - 1} \alpha_{n - k - j + 1} \, S_{s - j}^{k - 1} \br{n - j}
    .
    &&
\end{align}
Then the $m$-th member of the non-commutative discrete first Painlevé hierarchy has the form (see Theorem~\ref{thm:dPIm_nc})
\begin{align}
    \label{eq:dPIm_nc}
    \tag*{$\text{d-PI}_m^{\text{nc}}$}
    &&
    \sum_{r = 1}^{m + 1}
    d_{2(r - 1)} \, S_r^r \br{n + r - 1}
    + c_n
    &= 0
    ,
    &
    c_n
    := n - c_0
    .
    &&
\end{align}
Here the constant parameters $d_l$ belong to the center $\mathcal{Z}(\mathcal{R})$ of $\mathcal{R}$ and do not depend on $n$. In order to derive this form, some properties of the non-commutative Svinin polynomials have been used (see Lemma \ref{thm:spr_nc}). More details can be found in Subsection~\ref{sec:dPIm_nc}, whereas the commutative case is considered in Subsection~\ref{sec:dPIm}. 

\begin{rem}
\label{rem:pneq}
In the commutative case, the compatibility condition can be reduced to a single equation, by eliminating either $a_n$ or $b_n$. A non-commutative analogue of this equation, when, for instance, the function $a_n$ is eliminated, reads
\begin{align}
    \alpha_n \alpha_{n+1} p_{n+2}
    - \alpha_n p_{n + 1} \br{
    \lambda^2 - \alpha_n
    }
    + \br{
    \lambda^2 - \alpha_n
    } p_n \alpha_n
    - p_{n - 1} \alpha_{n - 1} \alpha_n
    + 2 \alpha_n
    + \lambda \br{
    b_n \alpha_n
    - \alpha_n b_n
    }
    &= 0
    ,
\end{align}
where $p_n := a_n \alpha_n^{-1}$. Due to the last term, which is non-zero in the general case, we cannot eliminate the function $b_n$ and, therefore, need to solve a system of equations instead of a single equation. Thus, we devote Subsection \ref{sec:dPIm} to obtaining a solution of the commutative system to make the paper self-contained. Let us also note that the solution of the latter equation gives a constant $\tilde c_n$, while in our case we have just $c_n$. However, according to Remark \ref{rem:n2sol_nc}, the alternating part can be recovered.
\end{rem}

Thanks to the main theorem, Theorem \ref{thm:dPIm_nc}, we present three first members \ref{eq:dPI1_nc}, \ref{eq:dPI2_nc}, \ref{eq:dPI3_nc} of the \ref{eq:dPIm_nc} hierarchy and study their ``continuous'' analogues by taking the appropriate limits (see Section \ref{sec:dPI1_3}). It~turns out that they can be mapped to the non-commutative versions of the corresponding members of the first Painlevé hierarchy \ref{eq:PIm_nc}, briefly reviewed in Subsection~\ref{sec:PIm}. Recall that the first Painlevé hierarchy is defined with the help of the Lenard operator, which has a non-local term in the non-commutative setting. In~Proposition \ref{thm:Lenard_nc} and Lemma \ref{thm:Gn_rec}, we present more convenient expressions for their derivations.

Section \ref{sec:voltlat} is dedicated to the stationary reduction of the non-commutative Volterra lattice hierarchy \ref{eq:VLr_nc} to the \ref{eq:dPIm_nc}, while the final section, Section \ref{sec:disc}, discusses several open problems. In particular, the continuous limit in both commutative and non-commutative settings, the construction of another discrete Painlevé hierarchy associated with~\eqref{eq:dPI10_nc} or with its matrix version \eqref{eq:dPI10_mat}, and connections with other known non-commutative Painlevé equations. In particular, regarding the discrete Painlevé II case, the starting example, denoted~\ref{eq:dPII1_nc}, is presented there. A detailed study of these questions is postponed to future work.

\subsection*{Acknowledgments}
The author is deeply grateful to Giovanni Felder for fruitful discussions at the conference \href{https://swissmaprs.ch/events/around-associators-from-integrable-systems-to-number-theory/}{\textit{Around associators: from integrable systems to number theory}}, as these discussions inspired the present work.

\section{A hierarchy}

Consider a non-commutative analogue \eqref{eq:lsys_nc} of the linear problem associated with the commutative discrete first Painlevé equation \ref{eq:dPI1}. The compatibility condition then leads to the identity
\begin{multline}
    \br{
    \alpha_n^{-1} b_{n - 1} + \lambda \alpha_n^{-1} a_n \alpha_n^{-1}
    } w_{n - 1}
    + \br{
    \alpha_n^{-1} a_{n - 1} - \lambda \alpha_n^{-1} b_n
    - \lambda^2 \alpha_n^{-1} a_n \alpha_n^{-1}
    - \alpha_n
    } \, w_n
    \\
    = \br{
    b_{n + 1} \alpha_n^{-1} 
    + \lambda a_{n + 1} \alpha_{n + 1} \alpha_n
    } \, w_{n - 1}
    + \br{
    a_{n + 1} \alpha_{n + 1} 
    - \lambda b_{n + 1} \alpha_n^{-1}
    - \lambda^2 a_{n + 1} \alpha_{n + 1}^{-1} \alpha_n^{-1}
    } \, w_{n}
\end{multline}
obtained by equating the two different evaluations of the linear system. 
Since the function $w_n$ is arbitrary, equating the coefficients of $w_{n - 1}$ and $w_n$ yields the system 
\begin{align}
    \label{eq:compcond_nc}
    \left\{
    \begin{array}{rcl}
    \lambda^2 \br{
    p_n
    - \alpha_n p_{n + 1} \alpha_n^{-1}
    }
    + \lambda \br{
    b_n - \alpha_n b_{n + 1} \alpha_n^{-1}
    }
    + \alpha_n p_{n + 1}
    - p_{n - 1} \alpha_{n - 1}
    + 1
    &=& 0,  
    \\[2mm]
    \lambda \br{
    p_n
    - \alpha_n p_{n + 1} \alpha_n^{-1}
    }
    - \alpha_n b_{n + 1} \alpha_n^{-1}
    + b_{n - 1}
    &=& 0,
    \end{array}
    \right.
\end{align}
where we have used the notation $p_n := a_n \alpha_n^{-1}$.
Note that the first equation of \eqref{eq:compcond_nc} can be simplified with the help of the second one. Thus, we obtain the system
\begin{align}
    \label{eq:ccond_nc}
    \left\{
    \begin{array}{rlr}
    \lambda \br{
    \alpha_n p_{n + 1} - p_n \alpha_n
    }
    + \alpha_n b_{n + 1}
    - b_{n - 1} \alpha_n
    &=& 0,  
    \\[2mm]
    \lambda \br{
    b_{n - 1} - b_n
    }
    + p_{n - 1} \alpha_{n - 1}
    - \alpha_n p_{n + 1}
    - 1
    &=& 0
    .
    \end{array}
    \right.
\end{align}

Let $P_{i, n}$ and $B_{i, n}$ be elements of $\mathcal{R}$. We assume that $p_n$ and $b_n$ admit the expansions
\begin{align} 
    \label{eq:bpser_nc}
    &&
    p_n
    &= \sum_{k = 0}^{m} P_{2k, n} \lambda^{2k}
    ,
    &&&
    b_n
    &= \sum_{l = 0}^{2m + 1} B_{l, n} \lambda^l
    = \sum_{k = 0}^{m} B_{2k, n} \lambda^{2k}
    + \sum_{k = 0}^{m} B_{2k + 1, n} \lambda^{2k + 1}
    .
    &&
\end{align}
Substituting them into \eqref{eq:ccond_nc}, we collect the coefficients of $\lambda$. 
Since $\lambda$ is an arbitrary parameter and the compatibility condition must be satisfied, 
all coefficients of powers of $\lambda$ must vanish. As a result, we obtain the following system of equations 
for $B_{i,n}$ and $P_{i,n}$:
\begin{align}
    \label{eq:cond1_1_nc}
    \alpha_n B_{2k, n + 1} - B_{2k, n - 1} \alpha_n
    &= 0,
    \\
    \label{eq:cond1_2_nc}
    \alpha_n P_{2k, n + 1}
    - P_{2k, n} \alpha_n
    + \alpha_n B_{2k + 1, n + 1}
    - B_{2k + 1, n - 1} \alpha_n
    &= 0,
    \\[2mm]
    \label{eq:cond2_1_nc}
    B_{2m + 1, n - 1} - B_{2m + 1, n}
    &= 0,
    \\
    \label{eq:cond2_2_nc}
    B_{2k, n - 1} - B_{2k, n}
    &= 0,
    \qquad
    0 \leq k \leq m,
    \\
    \label{eq:cond2_3_nc}
    B_{2l - 1, n - 1} - B_{2l - 1, n}
    - \alpha_n P_{2l, n + 1} + P_{2l, n - 1} \alpha_{n - 1}
    &= 0,
    \qquad 
    1 \leq l \leq m
    ,
    \\
    \label{eq:cond2_4_nc}
    P_{0, n - 1} \alpha_{n - 1}
    - \alpha_n P_{0, n + 1} - 1
    &= 0.
\end{align}

\medskip
We aim to solve this system.
Note that equations \eqref{eq:cond1_1_nc}, \eqref{eq:cond2_1_nc}, and \eqref{eq:cond2_2_nc} 
can be solved immediately. We assume throughout the paper that $c_i$ and $d_i$ belong to the center $\mathcal{Z}(\mathcal{R})$\footnote{One may also consider that these parameters are non-commutative, but we are not interested in that case here.} and do not depend on $n$. 
Equations \eqref{eq:cond2_1_nc} and \eqref{eq:cond2_2_nc} lead to a solution independent of $n$:
\begin{align}
    B_{2m + 1, n}
    &= c_{2 m + 1},
    &
    B_{2k, n}
    &= c_{2k},
    \qquad
    0 \leq k \leq m
    .
\end{align}
Then equation \eqref{eq:cond1_2_nc} for $k = m$ becomes
\begin{align}
    \alpha_n P_{2m, n+1}
    - P_{2m, n} \alpha_n
    = 0 
    .
\end{align}
Assuming that $P_{2m, n} \in \mathcal{Z}(\mathcal{R})$, we obtain $P_{2m, n} = P_{2m, n + 1}$. We set $P_{2m, n} = d_{2m}$. The remaining equations~are
\begin{align}
    \label{eq:cond_2_nc}
    B_{2k - 1, n - 1} 
    - B_{2k - 1, n}
    - \alpha_n P_{2k, n + 1} + P_{2k, n - 1} \alpha_{n - 1}
    &= 0,
    \\
    \label{eq:cond_1_nc}
    \alpha_n P_{2(k - 1), n + 1}
    - P_{2(k - 1), n} \alpha_n
    + \alpha_n B_{2k - 1, n + 1}
    - B_{2k - 1, n - 1} \alpha_n
    &= 0
    ,
    \qquad 
    1 \leq k \leq m
    ;
    \\[2mm]
    \label{eq:p0_nc}
    P_{0, n - 1} \alpha_{n - 1}
    - \alpha_n P_{0, n + 1} - 1
    &= 0
    .
\end{align}
After solving the recursive system of equations \eqref{eq:cond_2_nc} -- \eqref{eq:cond_1_nc} 
and substituting its solution into \eqref{eq:p0_nc}, equation~\eqref{eq:p0_nc} yields a 
condition on the function $\alpha_n$. We refer to this condition as a \textit{non-commutative analogue of the $m$-th member of the discrete first Painlevé hierarchy}.

\medskip
In the commutative case, the system \eqref{eq:ccond_nc} can be reduced to a single equation for $p_n$ (see Remark \ref{rem:pneq}). 
By~substituting the series into this equation, one can similarly collect the coefficients of $\lambda$ 
to obtain a system of equations solely for $P_{i,n}$. This solution was derived in \cite{cresswell1999discrete}. 
Since, in the non-commutative case, we cannot eliminate $a_n$ from the system \eqref{eq:ccond_nc}, 
we instead deal with the system \eqref{eq:cond_2_nc} -- \eqref{eq:cond_1_nc} for $B_{i,n}$ and $P_{i,n}$. 
To~be self-contained, we first solve this system in the commutative setting (see Subsection \ref{sec:dPIm}), 
and then proceed with the non-commutative case considered in Subsection \ref{sec:dPIm_nc}.

\subsection{Commutative case}
\label{sec:dPIm}
Here $\alpha_n$, $P_{i,n}$, and $B_{i,n}$ commute with each other, and equations \eqref{eq:cond_2_nc} -- \eqref{eq:cond_1_nc} can be simplified to the form
\begin{align}
    \label{eq:cond_2}
    B_{2k - 1, n - 1} 
    - B_{2k - 1, n}
    - \alpha_n P_{2k, n + 1} + P_{2k, n - 1} \alpha_{n - 1}
    &= 0,
    \\
    \label{eq:cond_1}
    P_{2(k - 1), n + 1}
    - P_{2(k - 1), n}
    + B_{2k - 1, n + 1}
    - B_{2k - 1, n - 1}
    &= 0,
    \qquad 
    1 \leq k \leq m
    ;
    \\[2mm]
    \label{eq:p0}
    \alpha_{n - 1} P_{0, n - 1} 
    - \alpha_n P_{0, n + 1} - 1
    &= 0
    .
\end{align}

\begin{prop}
\label{thm:bpsol}
The system \eqref{eq:cond_2} -- \eqref{eq:cond_1} has the following solution
\begin{align}
    \label{eq:psol}
    &P_{2(k - 1), n}
    = \br{d_{2(k - 1)} 
    + 2 c_{2k - 1}}
    - B_{2k - 1, n}
    - B_{2k - 1, n - 1},
    \\[2mm]
    \label{eq:bsol}
    &
    \begin{aligned}
    B_{2k - 1, n}
    = c_{2k - 1} 
    - d_{2(k-1)} \, \alpha_n
    &+ \, \alpha_n \br{
    B_{2k + 1, n + 1}
    + B_{2k + 1, n}
    - 2 c_{2k + 1}
    }
    \\
    &+ \,\sum_{j \leq n} \alpha_{j - 1} \, 
    B_{2k + 1, j}
    - \sum_{j \leq n - 1} \alpha_{j} \, 
    B_{2k + 1, j - 1}
    ,
    \quad 
    1 \leq k \leq m
    .
    \end{aligned}
\end{align}
\end{prop}
\begin{proof}
The proof proceeds by introducing an auxiliary function $A_{k, n}$ such that a difference equation can be rewritten in the form $A_{k, n + 1} - A_{k, n} = 0$. Hence, $A_{k, n + 1} = A_{k, n}$, and we can set $A_{k, n} = \const$.

\medskip
\textbullet \,\,
In order to solve \eqref{eq:cond_1}, we take the function
\begin{align}
    D_{k, n}
    &:= P_{2 (k - 1), n} + B_{2k-1, n} + B_{2k-1, n - 1}
\end{align}
and consider the difference $D_{k, n + 1} - D_{k, n}$. Then, 
\begin{align}
    D_{k, n + 1}
    - D_{k, n}
    &= P_{2 k - 2, n + 1} + B_{2k-1, n + 1} + B_{2k-1, n} 
    - \br{P_{2k-2, n} + B_{2k-1, n} + B_{2k-1, n - 1}}
    \\[2mm]
    &= P_{2k-2, n + 1} - P_{2k-2, n}
    + B_{2k-1, n + 1} - B_{2k-1, n - 1}
    = 0
    ,
\end{align}
due to the fact that \eqref{eq:cond_1} holds. Hence, $D_{k, n + 1} = D_{k, n}$, and $D_{k, n} = \tilde{d}_{2(k-1)} := d_{2(k-1)} + 2c_{2k-1}$.

\medskip
\textbullet \,\,
We now solve the second set of equations \eqref{eq:cond_2} in a similar manner. Define the function
\begin{align}
    C_{k, n}
    &:= B_{2k - 1, n} 
    + \alpha_n P_{2k, n + 1} 
    + \sum_{j \leq n - 1} \alpha_j \, \br{
    P_{2k, j + 1} - P_{2k, j}
    }
\end{align}
and consider $C_{k, n} - C_{k, n - 1}$. Thus, we obtain
\begin{align}
    C_{k, n - 1}
    &- C_{k, n}
    \\
    &= B_{2k - 1, n - 1} 
    + \alpha_{n - 1} P_{2k, n} 
    + \sum_{j \leq n - 2} \alpha_j \, \br{
    P_{2k, j + 1} - P_{2k, j}
    }
    \\&\phantom{= B_{2k - 1, n - 1}\,\,}
    - B_{2k - 1, n} 
    - \alpha_n P_{2k, n + 1} 
    - \sum_{j \leq n - 1} \alpha_j \, \br{
    P_{2k, j + 1} - P_{2k, j}
    }
    \\
    &= B_{2k - 1, n - 1} - B_{2k - 1, n}
    + \alpha_{n - 1} P_{2k, n}
    - \alpha_n P_{2k, n +1}
    - \alpha_{n - 1} \br{
    P_{2k, n} - P_{2k, n - 1}
    }
    = 0,
\end{align}
thanks to \eqref{eq:cond_2}. Hence, $C_{k, n - 1} = C_{k, n}$, and we set $C_{k, n} = c_{2k-1}$. Therefore, $B_{2k-1, n}$ reads as
\begin{align}
    B_{2k-1, n}
    &= c_{2k-1}
    - \alpha_n P_{2k, n + 1} 
    - \sum_{j \leq n - 1} \alpha_j \, \br{
    P_{2k, j + 1} - P_{2k, j}
    },
\end{align}
or as in \eqref{eq:bsol}, after substitution of \eqref{eq:psol} into that identity.
\end{proof}

The solution in Proposition \ref{thm:bpsol} is given by the nonlinear recurrence, which can be rewritten in the form
\begin{align}
    \label{eq:bsol_simp}
    &&&&
    \widetilde B_{0, n}
    &= 0,
    &&&
    \widetilde B_{k, n}
    &= d_k \alpha_n 
    + \alpha_n \, \widetilde B_{k - 1, n}
    + \sum_{j \leq n} \alpha_j \, \widetilde B_{k - 1, j + 1}
    - \sum_{j \leq n - 1} \, \alpha_j \widetilde B_{k - 1, j - 1}
    ,
    &
    k
    &= 1, \dots, m
    .
    &&&&
\end{align}
It turns out that one can present its explicit solution given in terms of the polynomials introduced in \cite{svinin2009some} and then studied in \cite{svinin2011some}, \cite{svinin2014some}. We will call these polynomials the \textit{Svinin polynomials} and denote them as $S_s^k(n)$. Namely, we have the following
\begin{prop}
\label{thm:bsolstr}
Let $S_{s}^k \br{n}$ be the Svinin polynomials defined recursively
\begin{align}
    \tag*{\eqref{eq:spoly}}
    &&
    S_{s}^0 \br{n}
    &= 1,
    &
    S_{s}^{k} \br{n}
    &= \sum_{j = 0}^{s - 1} \alpha_{n - k - j + 1} \, S_{s - j}^{k - 1} \br{n - j}
    .
    &&
\end{align}
Then the coefficients $\widetilde B_{k, n}$ given by \eqref{eq:bsol_simp} are expressed via the Svinin polynomials as follows
\begin{align}
    \label{eq:bsol_spoly}
    \widetilde B_{k, n}
    &= \sum_{r = 1}^{k} d_{k - r + 1} \, S_{r}^r \br{n + r - 1}
    ,
    \qquad
    0 \leq k \leq m
    .
\end{align}
In particular, $\deg \br{\widetilde B_{k, n}} = k$.
\end{prop}
Before we proceed with the proof, let us first list useful identities for the Svinin polynomials (see, e.g., Lemmas~4.1,~4.4, and 4.5 in~\cite{svinin2025volterra}):
\begin{gather}
    \label{eq:spr1}
    {S_s^k(n)}
    = S_{s-1}^k(n - 1) + \alpha_{n - k + 1} \, S_{s}^{k - 1}(n),
    \\[2mm]
    \label{eq:spr2}
    {S_s^k(n)}
    = S_{s-1}^k(n) + \alpha_{n - s + 1} \, S_{s}^{k - 1}(n - 1),
    \\[2mm]
    \label{eq:spr3}
    S_k^k (n) + S_k^k (n - 1)
    = S_{k + 1}^k (n)
    .
\end{gather}
Note also that these polynomials can be written explicitly in terms of $\alpha_n$:
\begin{align}
    \label{eq:spoly_prod}
    S_s^k(n)
    &= \sum_{0 \leq \lambda_1 \leq \dots \leq \lambda_k \leq s - 1} \, 
    \prod_{j = 1}^{k} 
    \alpha_{n - k - \lambda_j + j}
    .
\end{align}
\begin{rem}
The properties \eqref{eq:spr1} -- \eqref{eq:spr2} are just consequences of the definition of the Svinin polynomials \eqref{eq:spoly}, while \eqref{eq:spr3} follows from \eqref{eq:spr1} and the factorization identity \cite{svinin2015integrals}
\begin{align}
    \label{eq:spr_fact}
    &&
    S_s^k(n)
    &= S_{k + 1}^{s - 1} (n) \, 
    \prod_{j = s - 1}^{k - 1} \alpha_{n - j}
    ,
    &
    s
    &= 1, \dots, k
    .
    &&
\end{align}
\end{rem}
\begin{proof}[Proof of Proposition {\rm\ref{thm:bsolstr}}]
We prove the statement by induction. We omit the tildes and start with the second part of the proposition. 

\medskip
\textbullet \,\,
The base is $\deg \br{B_{0, n}} = 0$. Suppose that the statement is true for all $k = 1, \dots, l$, i.e. $\deg \br{B_{k, n}} = k$. 
Consider $k = l + 1$, then according to \eqref{eq:bsol_simp} the degree is increased by 1, since $B_{l + 1, n}$ is obtained from $B_{l}$ by multiplying by $\alpha_n$, and, therefore, 
\begin{align}
    \deg \br{B_{l + 1, n}} 
    = \deg \br{B_{l, n}} + 1 
    = l + 1 
    .
\end{align}

\medskip
\textbullet\,\,
Since $B_{0,n} = 0$, we consider $k = 1$ as the base. According to the statement and the recurrence formula, $B_{1, n} = d_1 \alpha_n$. Suppose that the statement is true for all $k \leq l - 1$, i.e.
\begin{align}
    B_{l - 1, n}
    &= \sum_{r = 1}^{l - 1} d_{l - r} S_{r}^r (n + r - 1)
    .
\end{align}
Consider $k = l$. We want to obtain that
\begin{align}
    B_{l, n}
    &= \sum_{r = 1}^{l} d_{l - r + 1} S_r^r (n + r - 1)
    = d_l \, S_1^1 (n) 
    + \sum_{r = 1}^{l - 1} d_{l - r} \, S_{r + 1}^{r + 1} (n + r)
    .
\end{align}
Indeed, according to \eqref{eq:bsol_simp} and our assumption, we have
\begin{align}
    B_{l, n}
    &= \alpha_n \br{
    d_l + B_{l - 1, n} + B_{l - 1, n + 1}
    }
    + \sum_{j \leq n - 1} \alpha_{j} \, 
    B_{l - 1, j + 1}
    - \sum_{j \leq n - 1} \alpha_{j} \, 
    B_{l - 1, j - 1}
    \\
    &= d_l \, \alpha_n
    + \sum_{r = 1}^{l - 1} d_{l - r} \, \alpha_n S_r^r (n + r - 1)
    + \sum_{r = 1}^{l - 1} d_{l - r} \, \alpha_n S_r^r (n + r)
    \\
    &\qquad\qquad
    + \sum_{j \leq n - 1} \sum_{r = 1}^{l - 1} d_{l - r} \, \alpha_{j} S_r^r (j + r)
    - \sum_{j \leq n - 1} \sum_{r = 1}^{l - 1} d_{l - r} \, \alpha_{j} S_r^r (j + r - 2)
    .
\end{align}
Since $d_l \, \alpha_n = d_l \, S_1^1(n)$, we want to show that
\begin{align}
    \sum_{r = 1}^{l - 1} d_{l - r} \, S_{r + 1}^{r + 1} (n + r)
    &= \sum_{r = 1}^{l - 1} d_{l - r} \, \alpha_n S_r^r (n + r - 1)
    + \sum_{r = 1}^{l - 1} d_{l - r} \, \alpha_n S_r^r (n + r)
    \\
    &\qquad\qquad
    + \sum_{j \leq n - 1} \sum_{r = 1}^{l - 1} d_{l - r} \, \alpha_{j} S_r^r (j + r)
    - \sum_{j \leq n - 1} \sum_{r = 1}^{l - 1} d_{l - r} \, \alpha_{j} S_r^r (j + r - 2)
    .
\end{align}
It suffices to verify the identity coefficientwise with respect to $d_l$. So, we have
\begin{align}
    \label{eq:srr0}
    S_{r + 1}^{r + 1} (n + r)
    &= \alpha_n \, \Big({ S_r^r \br{n + r}
    + S_r^r \br{n + r - 1}}\Big)
    + \sum_{j \leq n - 1} \alpha_{j} \, S_r^r (j + r)
    - \sum_{j \leq n - 1} S_r^r (j + r - 2) \, \alpha_{j} 
    \\
    \label{eq:srr1}
    S_{r + 1}^{r + 1} (n + r)
    &= \alpha_n \, S_{r + 1}^r (n + r)
    + \sum_{j \leq n - 1} \alpha_j \, S_r^r (j + r)
    - \sum_{j \leq n - 1} S_r^r (j + r - 2) \, \alpha_j
    ,
\end{align}
where we used \eqref{eq:spr3} for $k = r$ and $n$ shifted by $r$. 

Now let us rewrite the sum in \eqref{eq:srr1} by using \eqref{eq:spr1} -- \eqref{eq:spr2}. Consider \eqref{eq:spr1} with $k = r + 1$, $s = r$, and $n = j + r$,
\begin{align}
    S_r^{r+1}(j + r)
    - S_{r - 1}^{r + 1} (j + r - 1)
    = \alpha_j \, S_r^r (j + r)
    ,
\end{align}
and \eqref{eq:spr2} for $k = r + 1$, $s = r$, $n = j + r - 1$,
\begin{align}
    S_r^{r + 1} (j + r - 1)
    - S_{r - 1}^{r + 1} (j + r - 1)
    &= S_r^r (j + r - 2) \, \alpha_j
    .
\end{align}
Taking the sum of their differences over $j$, we obtain
\begin{align}
    \sum_{j \leq n - 1} \Big(
    \alpha_j \, S_r^r(j + r)
    - S_r^r (j + r - 2) \, \alpha_j
    \Big)
    &= \sum_{j \leq n - 1} \Big(
    S_r^{r+1}(j + r) - S_r^{r + 1} (j + r - 1)
    \Big)
    = S_r^{r + 1} (n + r - 1)
    .
\end{align}
Then, \eqref{eq:srr1} takes the form
\begin{align}
    S_{r + 1}^{r + 1} (n + r)
    &= \alpha_n \, S_{r + 1}^r (n + r)
    + S_r^{r + 1} (n + r - 1),
\end{align}
which coincides with \eqref{eq:spr1} for $k = s = r + 1$ and $n$ shifted by $r$.
\end{proof}

Finally, we can state the main theorem of this subsection concerning the $m$-th member of the discrete first Painlevé hierarchy.
\begin{thm}
\label{thm:dPIm}
The $m$-th member {\rm\ref{eq:dPIm}} of the discrete first d-Painlevé hierarchy reads as
\begin{align}
    \label{eq:dPIm}
    \tag*{d-PI$_m$}
    \sum_{r = 1}^{m + 1}
    d_{2(r - 1)} \, S_r^r \br{n + r - 1}
    + \br{n - c_0}
    &= 0
    .
\end{align}
\end{thm}
\begin{proof}
Let us define the function $\gamma_n$
\begin{align}
    \label{eq:gam}
    \gamma_n
    := n
    + d_0 \alpha_n 
    - \alpha_n \br{
    B_{1, n + 1} + B_{1, n} - 2 c_1
    }
    - \sum_{j \leq n - 1} \alpha_j \br{
    B_{1, j + 1} - B_{1, j - 1}
    }
    ,
\end{align}
and consider the equation \eqref{eq:p0}.
Recall that $P_{0, n} = d_0 + 2 c_1 - B_{1, n} - B_{1, n - 1}$. Then, $\gamma_{n}$ can be rewritten as
\begin{align}
    \gamma_n
    = n
    &
    + \alpha_{n} \br{
    d_0 - B_{1, n + 1} - B_{1, n} + 2 c_1
    }
    + \sum_{j \leq n - 1} \alpha_j
    \br{d_0 - B_{1, j + 1} - B_{1, j} + 2 c_1}
    \\
    &- \sum_{j \leq n - 1} \alpha_j \br{
    d_0 - B_{1, j} - B_{1, j - 1} + 2 c_1
    }
    = n
     + \alpha_n P_{0, n + 1}
    + \sum_{j \leq n - 1} \alpha_j \br{
    P_{0, j + 1} - P_{0, j}
    },
\end{align}
and
\begin{align}
    \gamma_{n - 1}
    &- \gamma_{n}
    \\
    &= \br{n - 1} 
    + \alpha_{n - 1} P_{0, n}
    + \sum_{j \leq n - 2} \alpha_{j} \br{
    P_{0, j + 1} - P_{0, j}
    }
    - n
    - \alpha_{n} P_{0, n + 1}
    - \sum_{j \leq n - 1} \alpha_j \br{
    P_{0, j + 1} 
    - P_{0, j}
    }
    \\
    &= - 1 
    + \alpha_{n - 1} P_{0, n}
    - \alpha_{n} P_{0, n + 1}
    - \alpha_{n - 1} \br{
    P_{0, n} - P_{0, n - 1}
    }
    = - \alpha_n P_{0, n + 1} 
    + \alpha_{n - 1} P_{0, n - 1}
    - 1
    = 0,
\end{align}
due to \eqref{eq:p0}. Hence, $\gamma_{n - 1} = \gamma_n$, and we set $\gamma_n = c_0$. Therefore, we obtain the condition on $\alpha_n$
\begin{align}
    d_0 \, \alpha_n 
    - \alpha_n \br{
    B_{1, n + 1} + B_{1, n} - 2 c_1
    }
    - \sum_{j \leq n - 1} \alpha_j \br{
    B_{1, j + 1} - B_{1, j - 1}
    }
    &= c_0 - n
    ,
\end{align}
which, after taking into account Proposition \ref{thm:bsolstr}, can be rewritten in the desired form. 
\end{proof}

\begin{rem}
\label{rem:n2sol}
The equation $\widetilde \gamma_{n - 1} - \widetilde \gamma_n = 1$ implies the equation $\widetilde \gamma_{n - 1} - \widetilde \gamma_{n + 1} = 2$, which has the solution $\widetilde \gamma_n = - n + c_0 + c_1 (-1)^n$.
\end{rem}

\medskip
We conclude this subsection by presenting several members of the \ref{eq:dPIm} hierarchy, 
as well as the solution $B_{2k-1, n}$ for $1 \leq k \leq m$ given in Proposition \ref{thm:bsolstr}. 
\begin{exmp}
For $m = 1$, we have
\begin{align}
    &&
    B_{3,n}
    &= c_3
    ,
    &
    B_{1,n}
    &= c_1 - d_2 \alpha_n
    ;
    &&
\end{align}
\begin{align}
    \label{eq:dPI1}
    \tag*{d-PI$_1$}
    c_0 - n
    =
    d_0 \alpha_n
    + d_2 \br{
    \alpha_{n - 1} \alpha_n 
    + \alpha_n^2
    + \alpha_n \alpha_{n + 1}
    }
    .
\end{align}
Note that the coefficient of $d_2$ can be rewritten as follows
\begin{align}
    \alpha_{n - 1} \alpha_n + \alpha_n^2 + \alpha_n \alpha_{n + 1} 
    &= \br{\alpha_{n - 1} + \alpha_n} \alpha_n -  \alpha_n^2 +  \alpha_n \br{\alpha_n +  \alpha_{n + 1}}
    .
\end{align}
When $m = 2$, one obtains
\begin{align}
    &&
    B_{5, n}
    &= c_5
    ,
    &
    B_{3,n}
    &= c_3
    - d_4 \alpha_n
    ,
    &
    B_{1, n}
    &= c_1 
    - d_2 \alpha_n 
    - d_4 \br{
    \alpha_{n - 1} \alpha_n 
    + \alpha_n^2
    + \alpha_n \alpha_{n + 1}
    }
    ;
    &&
\end{align}
\begin{align}
    \label{eq:dPI2}
    \tag*{d-PI$_2$}
    c_0 - n
    =
    d_0 \alpha_n
    + d_2 \br{
    \alpha_{n - 1} \alpha_n 
    + \alpha_n^2 
    + \alpha_n \alpha_{n + 1}
    }
    + d_4 \Big(
    \alpha_{n - 2} \alpha_{n - 1} \alpha_{n}
    + \alpha_{n - 1} \alpha_n \alpha_{n + 1}
    + \alpha_n \alpha_{n + 1} \alpha_{n + 2}
    \Big.
    \\
    \notag
    \left.
    + \, \br{
    \alpha_{n - 1} + \alpha_n
    }^2 \alpha_n
    - \alpha_n^3
    + \alpha_n \br{
    \alpha_n + \alpha_{n + 1}
    }^2
    \right)
    .
\end{align}
For $m = 3$, one gets
\begin{gather}
    \begin{aligned}
    B_{7, n}
    &= c_7
    ,
    &&&&&
    B_{5, n}
    &= c_5 - d_6 \alpha_n
    ,
    &&&&&
    B_{3, n}
    &= c_3 - d_4 \alpha_n
    - d_6 \br{
    \alpha_{n - 1} \alpha_n 
    + \alpha_n^2
    + \alpha_n \alpha_{n + 1}
    }
    ,
    \end{aligned}
    \\[2mm]
    \begin{aligned}
    B_{1, n}
    = c_1 - d_2 \alpha_n
    - d_4 \br{
    \alpha_{n - 1} \alpha_n 
    + \alpha_n^2
    + \alpha_n \alpha_{n + 1}
    }
    - d_6 \left(
    \alpha_{n - 2} \alpha_{n - 1} \alpha_{n}
    + \alpha_{n - 1}^2 \alpha_n
    + 2 \alpha_{n - 1} \alpha_n^2
    + \alpha_n^3
    \right.
    \\
    \left.
    + \, \alpha_{n - 1} \alpha_n \alpha_{n + 1}
    + 2 \alpha_n^2 \alpha_{n + 1}
    + \alpha_n \alpha_{n + 1}^2
    + \alpha_n \alpha_{n + 1} \alpha_{n + 2}
    \right)
    ;
    \end{aligned}
\end{gather}
\begin{align}
    \label{eq:dPI3}
    \tag*{d-PI$_3$}
    c_0 - n
    =
    d_0 \alpha_n
    + d_2 \br{
    \alpha_{n - 1} \alpha_n 
    + \alpha_n^2
    + \alpha_n \alpha_{n + 1}
    }
    + d_4 \Big(
    \alpha_{n - 2} \alpha_{n - 1} \alpha_n
    + \alpha_{n - 1} \alpha_n \alpha_{n + 1}
    + \alpha_n \alpha_{n + 1} \alpha_{n + 2}
    \Big.
    \\
    \notag
    \left.
    + \, \br{\alpha_{n - 1} + \alpha_n}^2 \alpha_n 
    - \alpha_n^3
    + \alpha_n \br{\alpha_n + \alpha_{n + 1}}^2
    \right)
    + d_6 \Big(
    \alpha _{n-3} \alpha _{n-2} \alpha _{n-1} \alpha _n
    + \alpha_{n - 2} \alpha_{n - 1} \alpha_n \alpha_{n + 1}
    \Big.
    \\
    + \, \alpha_{n - 1} \alpha_n \alpha_{n + 1} \alpha_{n + 2}
    + \alpha_n \alpha_{n + 1} \alpha_{n + 2} \alpha_{n + 3}
    + \alpha _{n-1} \alpha _n \alpha _{n-2}^2
    + 2 \alpha _{n-1} \alpha _n^2 \alpha _{n-2}
    + 2 \alpha _{n-1}^2 \alpha _n \alpha _{n-2}
    \\
    + \, \alpha _{n-1} \alpha _n \alpha _{n+1}^2
    + \alpha _n \alpha _{n+1} \alpha _{n+2}^2
    + 4 \alpha _{n-1} \alpha _n^2 \alpha _{n+1}
    + \alpha _{n-1}^2 \alpha _n \alpha _{n+1}
    + 2 \alpha _n \alpha _{n+1}^2 \alpha _{n+2}
    \\
    \left.
    + \, 2 \alpha _n^2 \alpha _{n+1} \alpha _{n+2}
    + \br{\alpha_{n - 1} + \alpha_n}^3 \alpha_n 
    - \alpha_n^4
    + \alpha_n \br{\alpha_n + \alpha_{n + 1}}^3
    \right)
    .
\end{align}
\end{exmp}

\subsection{Non-commutative case}
\label{sec:dPIm_nc}

Thanks to the fact that the commutative solution of the system is expressed in terms of the Svinin polynomials, we have found a solution of the non-commutative system in terms of a non-commutative analogue of the Svinin polynomials. Note that these polynomials appeared in \cite{casati2020recursion} (see~also~\cite{carpentier2022quantisations}), where the authors discovered a recursion operator related to a non-commutative analogue of the Volterra hierarchy. Here we assume that only $\alpha_n$ belongs to $\mathcal{R}$.

\begin{defn}
\label{def:spoly_nc}
A \textit{non-commutative analogue of the Svinin polynomials $S_{s}^k \br{n}$} is defined by
\begin{align}
    \label{eq:spoly_nc}
    &&
    S_{s}^0 \br{n}
    &= 1,
    &
    S_{s}^{k} \br{n}
    &= \sum_{j = 0}^{s - 1} \alpha_{n - k - j + 1} \, S_{s - j}^{k - 1} \br{n - j}
    .
    &&
\end{align}
\end{defn}

Note that according to the definition, one can obtain the following explicit formula
\begin{align}
    \label{eq:spoly_nc_prod}
    S_s^k(n)
    &= \sum_{0 \leq \lambda_1 \leq \dots \leq \lambda_k \leq s - 1} \, 
    \prod_{j = 1}^{k} 
    \alpha_{n - k - \lambda_j + j}
    ,
\end{align}
where the product of non-commutative elements is ordered as follows
\begin{align}
    \prod_{j = 1}^l \alpha_j 
    := \alpha_1 \, \alpha_2 \, \dots \, \alpha_l
    .
\end{align}
The case of $S_s^k(n)$ with $k = s$ and slightly modified index $n$ was considered in \cite{carpentier2022quantisations} (see~eq.~(13) therein). 
\medskip 

Below, we present an analogue of properties \eqref{eq:spr1} and \eqref{eq:spr2} (cf. with \cite{casati2020recursion}). The factorization property~\eqref{eq:spr_fact} no longer holds in the non-commutative setting. In addition, the identity \eqref{eq:spr3} is true only for $k = 1$. 
\begin{lem}
\label{thm:spr_nc}
Let $S_s^k(n)$ be the non-commutative Svinin polynomials. Then the following identities hold
\begin{gather}
    \label{eq:spr1_nc}
    {S_s^k(n)}
    = S_{s-1}^k(n - 1) + \alpha_{n - k + 1} \, S_{s}^{k - 1}(n),
    \\[2mm]
    \label{eq:spr2_nc}
    {S_{s}^k (n)}
    = S_{s - 1}^k (n) 
    + S_{s}^{k - 1} (n - 1) \, \alpha_{n - s + 1}
    .
\end{gather}
\end{lem}
\begin{proof} 
Both properties follow from the definition. 
\medskip

\textbullet \,\, Indeed, from the definition \eqref{eq:spoly_nc}, one has
\begin{align}
    S_{s}^k (n)
    &= \sum_{j = 0}^{s - 1} \alpha_{n - k - j + 1} \, 
    S_{s - j}^{k - 1} (n - j)
    = \alpha_{n - k + 1} \, S_{s}^{k - 1} (n)
    + \sum_{j = 1}^{s - 1} \alpha_{n - k - j + 1} \, 
    S_{s - j}^{k - 1} (n - j)
    \\
    &= \alpha_{n - k + 1} \, S_{s}^{k - 1} (n)
    + \sum_{j = 0}^{s - 2} \alpha_{(n - 1) - k - j + 1} \, 
    S_{s - j - 1}^{k - 1} (n - 1 - j)
    = \alpha_{n - k + 1} \, S_{s}^{k - 1} (n)
    + S_{s - 1}^{k} (n - 1)
    ,
\end{align}
which is nothing but \eqref{eq:spr1_nc}.
\medskip

\textbullet \,\, 
Now we use \eqref{eq:spoly_nc_prod}. Consider the difference
\begin{align}
    S_s^k(n) 
    &-\, S_{s - 1}^k (n)
    \\
    &= \sum_{0 \leq \lambda_1 \leq \dots \leq \lambda_k \leq s - 1} \, \prod_{j = 1}^{k} 
    \alpha_{n - k - \lambda_j + j}
    - \sum_{0 \leq \lambda_1 \leq \dots \leq \lambda_k \leq s - 2} \, \prod_{j = 1}^{k} 
    \alpha_{n - k - \lambda_j + j}
    \\
    &= \sum_{0 \leq \lambda_1 \leq \dots \leq \lambda_{k - 1} \leq s - 1} \prod_{j = 1}^{k - 1} 
    \alpha_{n - k - \lambda_j + j} \,\, 
    \alpha_{n - k - (s - 1) + k}
    \\
    &= \sum_{0 \leq \lambda_1 \leq \dots \leq \lambda_{k - 1} \leq s - 1} \prod_{j = 1}^{k - 1} 
    \alpha_{(n - 1) - (k - 1) - \lambda_j + j} \,\, 
    \alpha_{n - (s - 1)}
    = S_s^{k - 1} (n - 1) \,\, \alpha_{n - s + 1}
    ,
\end{align}
which gives us \eqref{eq:spr2_nc}.
\end{proof}

\begin{prop}
\label{thm:bpsol_nc}
A solution of the system \eqref{eq:cond_2_nc} -- \eqref{eq:cond_1_nc} is given by
\begin{align}
    \label{eq:bsol_nc}
    B_{2k + 1, n}
    &= c_{2k + 1}
    - \sum_{r = 1}^{m - k} d_{2(k + r)} \, 
    S_r^r (n + r - 1)
    ,
    \\
    \label{eq:psol_nc}
    P_{2 k, n}
    &= \sum_{r = 0}^{m - k} d_{2(k + r)} \, 
    S_{r + 1}^{r} (n + r - 1)
    ,
    \qquad
    0 \leq k \leq m - 1
    .
\end{align}
\end{prop}
\begin{rem}
Note that the formulas above are also valid for $k = m$.
\end{rem}
\begin{proof}
The proof is based on direct verification that \eqref{eq:bsol_nc} -- \eqref{eq:psol_nc} satisfy the non-commutative system. Before proceeding with this verification, we first make several useful observations.
\medskip

\textbullet \,\, Note that the equation \eqref{eq:cond_2_nc} can be solved in the same spirit as in Proposition \ref{thm:bpsol}. Indeed, we define the function $C_{k, n}$ as follows
\begin{align}
    C_{k, n}
    := B_{2k - 1, n} 
    + \alpha_n P_{2k, n + 1} 
    + \sum_{j \leq n - 1} \br{
    \alpha_j \, P_{2k, j + 1}
    - P_{2k, j} \, \alpha_j
    }
    .
\end{align}
Then $C_{k, n} - C_{k, n - 1} = 0$, thanks to \eqref{eq:cond_2_nc}:
\begin{align}
    C_{k, n}
    - C_{k, n - 1}
    &= B_{2k - 1, n} 
    - B_{2k - 1, n - 1} 
    + \alpha_n P_{2k, n + 1} 
    - \alpha_{n - 1} P_{2k, n} 
    \\
    &\qquad 
    + \, \sum_{j \leq n - 1} \br{
    \alpha_j \, P_{2k, j + 1}
    - P_{2k, j} \, \alpha_j
    }
    - \sum_{j \leq n - 2} \br{
    \alpha_j \, P_{2k, j + 1}
    - P_{2k, j} \, \alpha_j
    }
    \\
    &= B_{2k - 1, n} 
    - B_{2k - 1, n - 1} 
    + \alpha_n P_{2k, n + 1} 
    - \alpha_{n - 1} P_{2k, n} 
    + \br{
    \alpha_{n - 1} \, P_{2k, n}
    - P_{2k, n - 1} \, \alpha_{n - 1}
    } 
    \\
    &\qquad 
    + \, \sum_{j \leq n - 2} \br{
    \alpha_j \, P_{2k, j + 1}
    - P_{2k, j} \, \alpha_j
    }
    - \sum_{j \leq n - 2} \br{
    \alpha_j \, P_{2k, j + 1}
    - P_{2k, j} \, \alpha_j
    }
    \\
    &= B_{2k - 1, n} 
    - B_{2k - 1, n - 1} 
    + \alpha_n \, P_{2k, n + 1} 
    - P_{2k, n - 1} \, \alpha_{n - 1}
    = 0
    .
\end{align}
Hence, $C_{k, n} = C_{k, n - 1}$. We then set $C_{k, n} = c_{2k - 1}$. As a result, we obtain
\begin{align}
    \label{eq:bsol_nc0}
    &B_{2k - 1, n}
    = c_{2k - 1} 
    - \alpha_n P_{2k, n + 1}
    - \sum_{j \leq n - 1} \br{
    \alpha_j \, P_{2k, j + 1}
    - P_{2k, j} \alpha_j
    }
    ,
    \quad 
    1 \leq k \leq m
    .
\end{align}

\textbullet \,\, 
Using \eqref{eq:bsol_nc0}, we can eliminate $B_{l, n}$ from \eqref{eq:cond_1_nc}:
\begin{align}
    \alpha_n P_{2k - 2, n + 1}
    &- \, P_{2k - 2, n} \alpha_n
    \\
    = &-\, \alpha_n \, \br{
    c_{2k - 1} 
    - \alpha_{n + 1} P_{2k, n + 2}
    - \sum_{j \leq n} \br{
    \alpha_j \, P_{2k, j + 1}
    - P_{2k, j} \alpha_j
    }
    }
    \\
    &+ \, \br{c_{2k - 1} 
    - \alpha_{n - 1} P_{2k, n}
    - \sum_{j \leq n - 2} \br{
    \alpha_j \, P_{2k, j + 1}
    - P_{2k, j} \alpha_j
    }
    } \, \alpha_n
    \\[2mm]
    = & \, \orange{\alpha_n \alpha_{n + 1} P_{2k, n + 2}
    + \alpha_n \, \br{
    \alpha_n P_{2k, n + 1}
    + \alpha_{n - 1} P_{2k, n}
    + \alpha_{n - 2} P_{2k, n - 1}
    + \dots
    }}
    \\
    &- \, \alpha_n \br{
    P_{2k, n} \alpha_n 
    + P_{2k, n - 1} \alpha_{n - 1}
    + P_{2k, n - 2} \alpha_{n - 2}
    + \dots
    } \mseagreen{- \, \alpha_{n - 1} P_{2k, n} \alpha_n}
    \\
    &\mseagreen{- \,\br{
    \alpha_{n - 2} P_{2k, n - 1}
    + \alpha_{n - 3} P_{2k, n - 2}
    + \dots
    } \, \alpha_n}
    + \br{
    P_{2k, n - 2} \alpha_{n - 2}
    + P_{2k, n - 3} \alpha_{n - 3}
    + \dots
    } \, \alpha_n
    \\[2mm]
    =&\,
    \alpha_n \, \orange{\sum_{j \leq n} \alpha_{j + 1} P_{2k, j + 2}}
    - \alpha_n \, \sum_{j \leq n} P_{2k, j} \alpha_j
    - \mseagreen{\sum_{j \leq n} \alpha_{j - 1} P_{2k, j}} \, \alpha_n
    + \sum_{j \leq n} P_{2k, j - 2} \alpha_{j - 2} \, \alpha_n
    \\[2mm]
    =&\, \alpha_n \, \sum_{j \leq n} \br{
    \alpha_{j + 1} P_{2k, j + 2} 
    - P_{2k, j} \alpha_j
    } - \sum_{j \leq n - 1} \br{
    \alpha_{j} P_{2k, j + 1}
    - P_{2k, j - 1} \alpha_{j - 1}
    } \, \alpha_n
    .
\end{align}
Thus, one obtains a condition involving $P_{l, n}$ only:
\begin{align}
    \label{eq:psol_nc0}
    \begin{aligned}
     \alpha_n P_{2k - 2, n + 1}
    &- \, P_{2k - 2, n} \alpha_n
    \\
    &= \alpha_n \, \sum_{j \leq n} \br{
    \alpha_{j + 1} P_{2k, j + 2} 
    - P_{2k, j} \alpha_j
    } - \sum_{j \leq n - 1} \br{
    \alpha_{j} P_{2k, j + 1}
    - P_{2k, j - 1} \alpha_{j - 1}
    } \, \alpha_n
    .        
    \end{aligned}
\end{align}

\medskip
\textbullet \,\, 
First, we show that \eqref{eq:psol_nc} solves \eqref{eq:psol_nc0}. Since the $d_l$ are arbitrary, it is sufficient to verify \eqref{eq:psol_nc0} coefficientwise. We rewrite the infinite sums using the properties \eqref{eq:spr1_nc} and \eqref{eq:spr2_nc} with $k = r + 1$, $s = r + 1$, after the shifts of $n$ by $r$ and $r - 1$, respectively,
\begin{align}
    \alpha_n \, S_{r+1}^r (n + r)
    &= S_{r+1}^{r+1}(n + r) - S_{r}^{r + 1} (n + r - 1),
    \\[2mm]
    S_{r + 1}^r (n + r - 2) \, \alpha_{n - 1}
    &= S_{r + 1}^{r + 1} (n + r - 1) - S_{r}^{r + 1} (n + r - 1)
    .
\end{align}
Thus, 
\begin{align}
    \alpha_j \, S_{r + 1}^r (j + r)
    - S_{r + 1}^r (j + r - 2) \, \alpha_{j - 1}
    &= S_{r + 1}^{r + 1} (j + r) 
    - S_{r + 1}^{r + 1} (j + r - 1)
    .
\end{align}
Summing over $j$, we obtain
\begin{align}
    \label{eq:ssum}
    \sum_{j \leq n - 1} \br{
    S_{r + 1}^{r + 1} (j + r) 
    - S_{r + 1}^{r + 1} (j + r - 1)
    }
    &= S_{r + 1}^{r + 1} (n + r - 1)
    .
\end{align}
Using this identity, the right-hand side of \eqref{eq:psol_nc0} can be rewritten as
\begin{align}
    \alpha_n \, S_{r + 1}^{r + 1} (n + r + 1)
    - S_{r + 1}^{r + 1} (n + r - 1) \, \alpha_n
    &= S_{r + 2}^{r + 2} (n + r + 1) 
    - S_{r + 1}^{r + 2} (n + r)
    ,
\end{align}
where we have applied \eqref{eq:spr1_nc} and \eqref{eq:spr2_nc} with $k = r + 2$, $s = r + 1$, and the proper shifts of $n$. 

Similarly, the left-hand side of \eqref{eq:psol_nc0} can be rewritten as
\begin{align}
    \alpha_n \, P_{2(k - 1), n + 1}
    - P_{2(k - 1), n} \, \alpha_n
    \sim S_{r + 1}^{r + 2} (n + r + 1)
    - S_{r + 1}^{r + 2} (n + r)
    ,
\end{align}
where $\sim$ denotes equality of the coefficient of $d_l$. Hence, \eqref{eq:psol_nc0} becomes an identity upon substituting \eqref{eq:psol_nc}.
\medskip

\textbullet \,\, 
Now we plug \eqref{eq:psol_nc} into \eqref{eq:bsol_nc0} in order to obtain \eqref{eq:bsol_nc}. The substitution gives us
\begin{align}
    B_{2k - 1, n}
    = c_{2k - 1}
    &- \alpha_n \, \sum_{r = 0}^{m - k} d_{2 (k + r)} \, S_{r + 1}^r (n + r)
    \\
    &- \sum_{j \leq n - 1} \sum_{r = 1}^{m - k} d_{2(k + r)} \, \br{
    \alpha_j \, S_{r + 1}^r (j + r) 
    - S_{r + 1}^r (j + r - 1) \, \alpha_j
    }
    .
\end{align}
The sum simplifies due to \eqref{eq:ssum}, and we obtain
\begin{align}
    B_{2k - 1, n}
    &= c_{2k - 1}
    - d_{2k} \, \alpha_n 
    - \sum_{r = 1}^{m - k} d_{2 (k + r)} \, \br{
    \alpha_n \, S_{r + 1}^r (n + r)
    + S_{r}^{r + 1} (n + r - 1)
    }.
\end{align}
Using \eqref{eq:spr1_nc} and relabeling the indices, this expression turns into \eqref{eq:bsol_nc}. 
\end{proof}

The second set of equations \eqref{eq:cond_1_nc} is more involved to solve due to the conjugation by $\alpha_n$. This leads to an additional term $D_{k, n}$ in the solution (see Proposition \ref{thm:psol_nc_com} below), which disappears in the commutative case.
\begin{lem}
\label{thm:bdcond_nc}
Let $D_{k, n} := - (d_{2k} + 2 c_{2k+1}) + P_{2k, n} + B_{2k+1, n} + B_{2k+1, n - 1}$, where $P_{2k, n}$ and $B_{2k + 1, n}$ satisfy Proposition~{\rm\ref{thm:bpsol_nc}}. Then the element $D_{k, n}$ satisfies the identity
\begin{align}
    \label{eq:bdcond_nc}
    \alpha_n \, B_{2k + 1, n} 
    - B_{2k + 1, n} \, \alpha_n 
    &= \alpha_n \, D_{k, n + 1} - D_{k, n} \, \alpha_n
    ,
\end{align}
and vanishes under the commutative reduction.
\end{lem}
\begin{proof} 
Since the constants $d_l$ are arbitrary, it suffices to prove the statement coefficientwise with respect to~$d_l$.
\medskip

\textbullet \,\, 
Let us start with the second part of the statement. We denote by $D_{r,n}$ the coefficient of $d_{2r}$ in $D_{k, n}$. By using Proposition \ref{thm:bpsol_nc}, its explicit form is given by
\begin{align}
    D_{r,n}
    &= S_{r + 1}^r (n + r) 
    - S_r^r (n + r) 
    - S_r^r (n + r - 1).
\end{align}
Under the commutative reduction the right-hand side is equal to zero, due to the property \eqref{eq:spr3}. Note that in the non-commutative case this property generally fails for $r > 1$ and survives only when $r = 1$. 

\medskip
\textbullet \,\,
The identity \eqref{eq:bdcond_nc} follows from the properties of the non-commutative Svinin polynomials given in Lemma \ref{thm:spr_nc}. Substituting the expressions for $P_{2k, n}$ and $B_{2k+1, n}$ from Proposition \ref{thm:bpsol_nc} into \eqref{eq:bdcond_nc}, the coefficient of $d_{2r}$ reads as
\begin{align}
    \alpha_n \, \br{
    S_{r + 1}^r (n + r) - S_r^r (n + r)
    }
    - \br{
    S_{r + 1}^r (n + r - 1) - S_r^r (n + r - 2)
    } \, \alpha_n
    .
\end{align}
The latter is equal to zero due to the properties of $S_s^k(n)$. Indeed, applying \eqref{eq:spr1_nc} with $k = r + 1$, $s = r + 1$, $n \mapsto n - r$, and subtracting the same identity with $s=r$ yields
\begin{align}
    \alpha_n \br{
    S_{r+1}^r (n + r) - S_r^r (n + r)
    }
    &= S_{r+1}^{r+1}(n + r) - S_r^{r + 1}(n + r)
    - S_{r}^{r + 1} (n + r - 1)
    + S_{r - 1}^{r + 1} (n + r - 1)
    ,
\end{align}
while the difference of \eqref{eq:spr2_nc} with $k = r + 1$, $s = r + 1$, $n \mapsto n - r$ and $k = r + 1$, $s = r$, $n \mapsto n - r + 1$ reads
\begin{align}
    \br{
    S_{r+1}^r (n + r - 1) - S_r^r (n + r - 2)
    }  \alpha_n
    &= S_{r+1}^{r+1}(n + r) - S_r^{r + 1}(n + r)
    - S_{r}^{r + 1} (n + r - 1)
    + S_{r - 1}^{r + 1} (n + r - 1)
    .
\end{align}
Hence, 
\begin{align}
    \alpha_n \, \br{
    S_{r + 1}^r (n + r) - S_r^r (n + r)
    }
    - \br{
    S_{r + 1}^r (n + r - 1) - S_r^r (n + r - 2)
    } \, \alpha_n
    = 0,
\end{align}
which completes the proof.
\end{proof}
Thanks to Lemma \ref{thm:bdcond_nc}, we have a non-commutative analogue of Proposition \ref{thm:bpsol}.
\begin{prop}
\label{thm:psol_nc_com}
The set of equations \eqref{eq:cond_1_nc} admits a solution of the form
\begin{align}
    \label{eq:psol_nc_com}
    &P_{2k, n}
    = \br{d_{2k} + 2 c_{2k + 1}}
    - B_{2k + 1, n}
    - B_{2k + 1, n - 1}
    + D_{k, n}
    ,
    \quad 
    0 \leq k \leq m - 1
    .
\end{align}
\end{prop}
\begin{proof}
This follows from the fact that \eqref{eq:cond_1_nc} can be rewritten as
\begin{align}
    \alpha_n \br{
    P_{2k, n + 1}
    - B_{2k + 1, n + 1}
    - B_{2k + 1, n}
    + B_{2k + 1, n}
    }
    - \br{
    P_{2k, n}
    - B_{2k + 1, n}
    - B_{2k + 1, n - 1}
    + B_{2k + 1, n}
    } \alpha_n
    = 0,
\end{align}
or, after using Lemma \ref{thm:bdcond_nc},
\begin{align}
    \alpha_n \br{
    P_{2k, n + 1}
    - B_{2k + 1, n + 1}
    - B_{2k + 1, n}
    + D_{k, n + 1}
    }
    - \br{
    P_{2k, n}
    - B_{2k + 1, n}
    - B_{2k + 1, n - 1}
    + D_{k, n}
    } \, \alpha_n
    = 0.
\end{align}
Therefore, we arrive at \eqref{eq:psol_nc_com}.
\end{proof}

\medskip
Finally, we formulate the main theorem on a non-commutative version of the \ref{eq:dPIm} hierarchy.
\begin{thm}
\label{thm:dPIm_nc}
The non-commutative analogue {\rm\ref{eq:dPIm_nc}} of the $m$-th member of the discrete first d-Painlevé hierarchy is given by
\begin{align}
    \tag*{\ref{eq:dPIm_nc}}
    \sum_{r = 1}^{m + 1}
    d_{2(r - 1)} \, S_r^r \br{n + r - 1}
    + \br{n - c_0}
    &= 0
    .
\end{align}
\end{thm}
\begin{proof}
The proof is similar to the commutative case (see Theorem \ref{thm:dPIm}). We define the function $\gamma_n$ as follows
\begin{align}
    \label{eq:gam_nc}
    \gamma_n
    := n
    + \alpha_n P_{0, n + 1}
    + \sum_{j \leq n - 1} \br{
    \alpha_j P_{0, j + 1}
    - P_{0, j} \alpha_j
    }
    .
\end{align}
Then $\gamma_{n - 1} - \gamma_n = 0$, since \eqref{eq:p0_nc} holds:
\begin{align}
    \gamma_{n - 1}
    &- \gamma_{n}
    \\
    &= \br{n - 1} 
    + \alpha_{n - 1} P_{0, n}
    + \sum_{j \leq n - 2} \br{
    \alpha_{j} P_{0, j + 1} - P_{0, j} \alpha_{j}
    }
    - n
    - \alpha_{n} P_{0, n + 1}
    - \sum_{j \leq n - 1} \br{
    \alpha_j P_{0, j + 1} 
    - P_{0, j} \alpha_j
    }
    \\
    &= - 1 
    + \alpha_{n - 1} P_{0, n}
    - \alpha_{n} P_{0, n + 1}
    - \br{
    \alpha_{n - 1} P_{0, n} 
    - P_{0, n - 1} \alpha_{n - 1} 
    }
    = - \alpha_n P_{0, n + 1} 
    + \alpha_{n - 1} P_{0, n - 1}
    - 1
    = 0.
\end{align}
Hence, $\gamma_{n - 1} = \gamma_n$, and we set $\gamma_n = c_0$. Considering the explicit form of the $P_{l, n}$ coefficients given in Proposition \ref{thm:bpsol_nc}, we arrive at the \ref{eq:dPIm_nc} equation.
\end{proof}

\begin{rem}
\label{rem:n2sol_nc}
In order to obtain the solution with an alternating contribution, one needs to use a consequence of the equation $\gamma_{n - 1} - \gamma_n = 1$ which is $\gamma_{n - 1} - \gamma_{n + 1} = 2$ (see also  Remark \ref{rem:n2sol}). 
\end{rem}

\subsection{Continuous cases}
\label{sec:PIm}
A commutative hierarchy of the first Painlevé equation was derived in \cite{kudryashov1997first}, by taking the stationary reduction of the Korteweg-de Vries hierarchy. The Korteweg-de Vries hierarchy has a matrix analogue (see, e.g., \cite{OS_1998a}), which was used by the authors of the paper \cite{Pic} in order to obtain a matrix version of the first Painlevé hierarchy. This hierarchy is defined with the help of the Lenard operator, which in the non-commutative setting is given by
\begin{align}
    \label{eq:Lenard_nc}
    \begin{aligned}
    &\begin{aligned}
    \ncLenard{0}
    &= \tfrac{1}{2},
    &&&&&&&&&
    \ncLen{l}[0]
    &= 0,
    &&&
    \forall \,\, l
    &\in \mathbb{N}
    ,
    \end{aligned}
    \\[2mm]
    &\begin{aligned}
    \partial \ncLenard{n + 1} 
    &= \br{
    \partial^3
    + \lbr{u, -}_+ \, \partial
    + \partial \, \lbr{u, -}_+
    + \lbr{u, -} \, \partial^{-1} \br{\lbr{u, -}}
    } \, \ncLenard{n}
    ,
    &
    n &\in \mathbb{N}
    .
    \end{aligned}
    \end{aligned}
\end{align}
Here $\lbr{-, -}_+$ is an anticommutator and $\partial^{-1}$ is a formal antiderivative. We assume that $\partial^{-1} (f) = g$ if $f = \partial (g)$. Note that locality of the non-local term in \eqref{eq:Lenard_nc} was proved in \cite{olver2000classification} (see Theorem 6.2 therein). 

The Lenard operator $\ncLen{n + 1}[u]$ in \eqref{eq:Lenard_nc} can be written explicitly in terms of $\ncLen{l}[u]$, $l \leq n$ and their derivatives, so that we get a more convenient expression for computations that is a generalization to the non-commutative case of the formula obtained in \cite{mazzocco2007hamiltonian} (see eq. (54) therein).
\begin{prop}
\label{thm:Lenard_nc}
Let $n \in \mathbb{N}$. The Lenard operator given by \eqref{eq:Lenard_nc} can be written as follows
\begin{align}
    &\begin{aligned}
    \ncLenard{0}
    &= \tfrac{1}{2},
    &&&&&
    \ncLenard{1}
    &= u,
    \end{aligned}
    \\[2mm]
    \label{eq:len_nc}
    &\begin{aligned}
    \ncLenard{n + 1} 
    &=
    {\ncLen{n}}[u]''
    + \tfrac{3}{2} \lbr{\ncLen{n}[u], \ncLen{1}[u]}_+
    \\
    &\phantom{=}
    + \, \tfrac{1}{2} 
    \sum_{j = 1}^{n - 1}
    \br{
        \lbr{
        \ncLen{n - j}[u],
            2 \lbr{\ncLen{1}[u], \ncLen{j}[u]}_+
            - \ncLen{j + 1}[u]
            + 2 {\ncLen{j}}[u]^{\prime\prime} 
        }_+
        - \lbr{{\ncLen{j}}[u]^{\prime},
        {\ncLen{n - j}}[u]^{\prime}}_+
    }.
    \end{aligned}
\end{align}
\end{prop}

The proof is based on formal integration by parts, certain relations between commutators and anti-commutators, and the lemma below. For the sake of readability, we moved the proof to Appendix~\ref{app:Lenard_nc}.
\begin{lem}
\label{thm:Gn_rec}
The non-local term $G_n [u] := \lbr{\ncLen{1}[u], \partial^{-1} \br{\lbr{\ncLen{1}[u], \ncLen{n}[u]}}}$ in \eqref{eq:Lenard_nc} satisfies the given formula
\begin{align}
    \label{eq:Gn_rec}
    &&
    G_n[u]
    &= \sum_{j = 1}^{n - 1} \br{
    \lbr{\ncLen{1}[u], \lbr{\ncLen{n - j}[u], {\ncLen{j}}[u]'}}
    - \lbr{\ncLen{n - j}[u], G_{j}[u]}_+
    }
    ,
    &
    n 
    &\in \mathbb{N}
    .
    &&
\end{align}
\end{lem}

\begin{exmp}
A few members of the $\ncLenard{l}$ are
\begin{align}
    &\begin{aligned}
    \ncLenard{2}
    &= u''
    + 3 u^2
    ,
    &&&&&
    \ncLenard{3}
    &= u^{(4)}
    + 5 u u''
    + 5 u'' u
    + 5 \br{u'}^2
    + 10 u^3
    ,
    \end{aligned}
    \\[2mm]
    &\begin{aligned}
    \ncLenard{4}
    = u^{(6)}
    + 7 u u^{(4)}
    + 7 u^{(4)} u
    &+ \, 21 \br{u''}^2
    + 21 u^2 u''
    + 21 u'' u^2
    + 28 u u'' u
    \\
    &+ \, 28 u \br{u'}^2
    + 28 \br{u'}^2 u
    + 14 u' u u'
    + 14 u' u^{(3)}
    + 14 u^{(3)} u'
    + 35 u^4
    .
    \end{aligned}
\end{align}
\end{exmp}
\begin{defn}
\label{def:PIm_nc}
A \textit{non-commutative version of the first Painlevé hierarchy} reads as
\begin{align}
    \label{eq:PIm_nc}
    \tag*{$\text{PI}_m^{\text{nc}}$}
    &&
    \ncLenard{m}
    + z
    &= 0
    ,
    &
    m \in \mathbb{N}
    .
    &&
\end{align}
\end{defn}
\begin{rem}
In the paper \cite{Pic}, the left-hand side contains the sum of the Lenard operators:
\begin{align}
    &&
    \ncLenard{m}
    + \sum_{j = 1}^{m - 1} c_j \, \ncLenard{j}
    + z
    &= 0
    ,
    &
    m \in \mathbb{N}
    .
    &&
\end{align}
\end{rem}

\begin{rem}
\label{rem:contlim}
The {\rm\ref{eq:dPIm_nc}} hierarchy maps to a $2m$-th order non-commutative  differential equation by taking the following change of variables with the small parameter $\varepsilon$:
\begin{gather}
    \label{eq:dPIm_to_PIm_nc}
    \begin{gathered}
    \begin{aligned}
    \alpha_n
    &= 1 + \varepsilon^2 \, u,
    &&&&
    z
    &= n \, \varepsilon,
    &&&&
    c_0 + n
    &= z + \br{-1}^{m - 1} 
    \begin{pmatrix}
        2 m + 1
        \\
        m
    \end{pmatrix} \varepsilon^{- (2m + 2)}
    \end{aligned}
    \\[2mm]
    \begin{aligned}
    d_{2m}
    &= \varepsilon^{- (2m + 2)}
    ,
    &&&&
    d_{2j}
    &= \br{-1}^j \, \tilde c_j \, \varepsilon^{- (2m + 2)}
    ,
    \quad
    j = 0, \dots, m - 1
    .
    \end{aligned}
    \end{gathered}
\end{gather}
We will show below, in Section \ref{sec:dPI1_3}, that a proper choice of the coefficients $\tilde c_j$ maps the members of the \ref{eq:dPIm_nc} to the corresponding members of the \ref{eq:PIm_nc}. However, the map between these two hierarchies has not been established yet even in the commutative case. In the author's opinion, in order to make this relation explicit, the \ref{eq:dPIm_nc} as well as the \ref{eq:dPIm} should be reformulated with the help of a discrete version of the recursion operator, analogous to \eqref{eq:Lenard_nc}.

\end{rem}

\section{Non-commutative examples}
\label{sec:dPI1_3}

Below we list the first three members, \ref{eq:dPI1_nc}, \ref{eq:dPI2_nc}, \ref{eq:dPI3_nc}, of the non-commutative first d-Painlevé hierarchy \ref{eq:dPIm_nc} and study their continuous limits. Recall that, in the non-commutative setting, by a continuous limit we mean the following change of variables with the commutative parameter $\varepsilon$
\begin{align}
    &&
    z
    &= \varepsilon \, n
    &&
\end{align}
supplemented by the substitutions
\begin{align}
    &&
    \alpha_n
    &= u,
    &
    \alpha_{n + k}
    &= u 
    + \, k \, \varepsilon u'
    + \, \tfrac12 k^2 \, \varepsilon^2 u''
    + O (\varepsilon^3)
    .
    &&
\end{align}
These substitutions must be chosen in such a way that the limit $\varepsilon \to 0$ exists (see also Remark \ref{rem:contlim}). 

\subsection{The first member}
\label{sec:dPI1_nc}
The coefficients of the series \eqref{eq:bpser_nc}, according to Proposition \ref{thm:bpsol_nc}, are given~by
\begin{align}
    &&
    B_{3, n}
    &= c_3
    ,
    &
    B_{1, n}
    &= c_1 - d_2 \alpha_n,
    &
    P_{2,n}
    &= d_2,
    &
    P_{0, n}
    &= d_0 + d_2 \br{
    \alpha_{n - 1} + \alpha_{n}
    };
    &&
\end{align}
and the first member of the hierarchy takes the form
\begin{align}
    \label{eq:dPI1_nc}
    \tag*{$\text{d-PI}_1^{\text{nc}}$}
    c_0 - n
    =
    d_0 \alpha_n
    + d_2 \br{
    \alpha_{n - 1} \alpha_n 
    + \alpha_n^2
    + \alpha_n \alpha_{n + 1}
    }
    .
\end{align}
\begin{rem}
\label{rem:adler}
The \ref{eq:dPI1_nc} equation with $\alpha_n \in \Mat_k(\mathbb{C})$ was obtained in the paper \cite{adler2020} (dP$_1^1$ therein) by making a symmetry reduction of the matrix Volterra lattice.
\end{rem}
\begin{rem}
\label{rem:dPI10_nc}
Another non-commutative version of the d-PI equation can be written in the form
\begin{align}
    \label{eq:dPI10_nc}
    &&
    \br{
    \alpha_{n + 1} + \alpha_n + \alpha_{n - 1}
    } \alpha_n
    + t \, \alpha_n
    &= c_n
    ,
    &
    c_n
    &= c_0 + n
    ,
    &&
\end{align}
where $t$ is a central element independent on $n$ and $c_0$ is a commutative constant. 
Its matrix version was obtained in the paper \cite{cassatella2014singularity}, by studying the singularity confinement test, and is connected with matrix orthogonal polynomials \cite{manas2021matrix}. It was also derived in the paper \cite{bobrova2024affine} by using Backlund transformations of a non-commutative fourth Painlevé system. 
\end{rem}

\begin{rem}
Recall that \eqref{eq:dPI10_nc} can be rewritten as the system
\begin{align}
    \tag*{d-PI$_1^0$}
    \label{eq:dPI10_sys_nc}
    &&
    &\left\{
    \begin{array}{rcl}
         \alpha_{n} + \alpha_{n - 1}
         &=&  - t + \beta_n
         ,
         \\[2mm]
         \beta_{n + 1} + \beta_{n}
         &=&  t + \alpha_n + c_n \, \alpha_n^{-1} 
    \end{array}
    \right.
    &&
\end{align}
One can verify that the commutator $I \br{\alpha_n, \beta_n} = \alpha_n \, \beta_n - \beta_n \, \alpha_n$ is preserved under this discrete dynamics (see also Proposition 3.4 in \cite{bobrova2025non}).
\end{rem}

\subsubsection{On the commutator}
Let us rewrite the obtained second-order equation \ref{eq:dPI1_nc} as a system of two first-order difference equations. Set $d_2 = 1$, $d_0 = t$ and define $\beta_{n} := \alpha_n + \alpha_{n - 1} + t$, then
$\alpha_n + \alpha_{n - 1} = - t + \beta_n$, and the \ref{eq:dPI1_nc} takes the form
\begin{align}
    \br{\alpha_{n + 1}
    + \alpha_{n}
    + t}
    + \alpha_n^{-1} \br{
    \beta_n - \alpha_n - t
    } \alpha_n
    &= c_n \alpha_n^{-1},
    &
    \beta_{n + 1}
    + \alpha_n^{-1} \beta_n \alpha_n
    &= t + \alpha_n + c_n \alpha_n^{-1},
\end{align}
where $c_n = c_0 - n$. Therefore, we obtain the system equivalent to \ref{eq:dPI1_nc}:
\begin{align}   
    \label{eq:dPI1_sys_nc}
    \left\{
    \begin{array}{rcl}
    \alpha_n + \alpha_{n - 1}
    &=& - t + \beta_n,
    \\[2mm]
    \beta_{n + 1} + \alpha_n^{-1} \beta_n \alpha_n
    &=& t + \alpha_n + c_n \alpha_n^{-1}
    .
    \end{array}
    \right.
\end{align}
\begin{prop}
Let $T \br{\alpha_n, \beta_n} = \br{\alpha_{n + 1}, \beta_{n + 1}}$ be defined by \eqref{eq:dPI1_sys_nc} and $\sigma_\beta \br{\alpha, \beta} := \br{\alpha_n, \alpha_n \beta_n \alpha_n^{-1}}$. Then $T$ acts on the element $I \br{\alpha_n, \beta_n} = \alpha_n \, \beta_n - \beta_n \, \alpha_n$ as follows
\begin{align}
    T\br{
    \alpha_n \beta_n - \beta_n \alpha_n
    }
    &= \alpha_n^{-1} \br{
    \alpha_n \beta_n - \beta_n \alpha_n
    } \alpha_n
    = \sigma_\beta \br{
    \alpha_n \beta_n - \beta_n \alpha_n
    }
    ,
\end{align}
and $T \sigma_{\beta} \neq \sigma_{\beta} T$.
\end{prop}
\begin{proof}
Indeed, the difference of the following equations
\begin{align}
    \alpha_{n + 1} + \alpha_{n}
    &= - t + \beta_{n + 1}
    &\Leftrightarrow&
    &
    \beta_{n + 1}\alpha_{n + 1}\beta_{n + 1}^{-1} 
    + \beta_{n + 1}\alpha_{n}\beta_{n + 1}^{-1}
    &= - t + \beta_{n + 1},
    \\[2mm]
    \beta_{n + 1} + \alpha_n^{-1} \beta_n \alpha_n
    &= t + \alpha_n + d_n \alpha_n^{-1}
    &\Leftrightarrow&
    &
    \alpha_n \beta_{n + 1} \alpha_n^{-1} 
    + \beta_n
    &= t + \alpha_n + d_n \alpha_n^{-1}
\end{align}
leads to
\begin{align}
    \alpha_{n + 1} + \alpha_{n}
    - \beta_{n + 1}\alpha_{n + 1}\beta_{n + 1}^{-1} 
    - \beta_{n + 1}\alpha_{n}\beta_{n + 1}^{-1}
    &= 0
    &\Leftrightarrow&
    &
    \alpha_{n + 1} \beta_{n + 1}
    - \beta_{n + 1} \alpha_{n + 1}
    + \alpha_n \beta_{n + 1}
    - \beta_{n + 1} \alpha_n
    &= 0
    ,
    \\[2mm]
    \beta_{n + 1} + \alpha_n^{-1} \beta_n \alpha_n
    - \alpha_n \beta_{n + 1} \alpha_n^{-1} 
    - \beta_n
    &= 0
    &\Leftrightarrow&
    &
    \beta_{n + 1} \alpha_n
    - \alpha_n \beta_{n + 1}
    - \beta_n \alpha_n
    + \alpha_n^{-1} \beta_n \alpha_n^2
    &= 0
    .
\end{align}
Their sum yields the condition
\begin{align}
    \alpha_{n + 1} \beta_{n + 1}
    - \beta_{n + 1} \alpha_{n + 1}
    &= \beta_n \alpha_n - \alpha_n^{-1} \beta_n \alpha_n^2
    = \alpha_n \br{\alpha_n^{-1} \beta_n \alpha_n}
    - \br{\alpha_n^{-1} \beta_n \alpha_n} \alpha_n
    = \alpha_n \sigma_\beta \br{\beta_n}
    - \sigma_{\beta} \br{\beta_n} \alpha_n
    .
\end{align}

The second statement, i.e. $T \sigma_{\beta} \neq \sigma_{\beta} T$, can be proved by a direct computation.
\end{proof}

\subsubsection{Continuous limit}

By taking the change with the commutative parameter $\varepsilon$:
\begin{align}
    \label{eq:dPI1toPI1}
    \alpha_n
    &= 1 + \varepsilon^2 u,
    &
    z
    &= \varepsilon \, n,
    &
    d_2
    &= \varepsilon^{-4},
    &
    d_0
    &= - 6 \varepsilon^{-4},
    &
    c_0 + n
    &= z + 3 \varepsilon^{-4}
    ,
\end{align}
the \ref{eq:dPI1_nc} turns, in the limit $\varepsilon\to0$, to the non-commutative analogue of the first Painlevé equation \cite{Balandin_Sokolov_1998}:
\begin{align}
    \label{eq:PI1_nc}
    \tag*{$\text{PI}_1^{\text{nc}}$}
    u''
    + 3 u^2 
    + z
    &= 0
    .
\end{align}

\begin{rem}
The same change \eqref{eq:dPI1toPI1} applied to the equation \eqref{eq:dPI10_nc} also yields \ref{eq:PI1_nc}. 
\end{rem}

\subsection{The second member}
\label{sec:dPI2_nc}
In this case, the coefficients of \eqref{eq:bpser_nc} are
\begin{gather}
    \begin{aligned}
    B_{5, n}
    &= c_5,
    &&&
    B_{3, n}
    &= c_3 - d_4 \alpha_n
    ,
    &&&
    B_{1, n}
    &= c_1 - d_2 \alpha_n 
    - d_4 \br{
    \alpha_{n - 1} \alpha_n
    + \alpha_n^2
    + \alpha_n \alpha_{n + 1}
    }
    ,
    \end{aligned}
    \\[2mm]
    \begin{aligned}
    P_{4, n}
    &= d_4,
    &&&
    P_{2, n}
    &= d_2 + d_4 \br{
    \alpha_{n - 1} + \alpha_n
    }
    ,
    \end{aligned}
    \\[2mm]
    P_{0, n}
    = d_0 + d_2 \br{
    \alpha_{n - 1} + \alpha_n
    }
    + d_4 \br{
    \alpha_{n - 2} \alpha_{n - 1}
    + \alpha_{n - 1} \alpha_n
    + \alpha_n \alpha_{n - 1}
    + \alpha_n \alpha_{n + 1}
    + \alpha_{n - 1}^2
    + \alpha_n^2
    };
\end{gather}
and the second member has the form
\begin{align}
    \label{eq:dPI2_nc}
    \tag*{$\text{d-PI}_2^{\text{nc}}$}
    c_0 - n
    =
    d_0 \alpha_n
    + d_2 \br{
    \alpha_{n - 1} \alpha_n 
    + \alpha_n^2
    + \alpha_n \alpha_{n + 1}
    }
    + d_4 \Big(
    \alpha _{n-2} \alpha _{n-1} \alpha _n
    + \alpha _{n-1} \alpha _n \alpha _{n+1}
    + \alpha _n \alpha _{n+1} \alpha _{n+2}
    \Big.
    \\
    \left.
    + \, \br{
    \alpha_{n - 1} + \alpha_n
    }^2 \alpha_n
    - \alpha_n^3
    + \alpha_n \br{
    \alpha_{n + 1} + \alpha_n
    }^2
    \right)
    .
\end{align}

\subsubsection{Continuous limit}

The same change \eqref{eq:dPI1toPI1} applied to \ref{eq:dPI2_nc} gives the \ref{eq:PI1_nc} with shifted $z$. In order to obtain the fourth-order equation, we consider
\begin{align}
    \alpha_n
    &= 1 + \varepsilon^2 \, u
    ,
    &
    z
    &= n \, \varepsilon,
    &
    d_4
    &= \varepsilon^{-6}
    &
    d_2
    &= - 10 \varepsilon^{-6},
    &
    d_0
    &= 30 \varepsilon^{-6},
    &
    c_0 + n
    &= z - 10 \varepsilon^{-6}
    .
\end{align}
Then, the limit $\varepsilon \to 0$ gives us a non-commutative analogue of the the second member of the first Painlevé hierarchy, the \ref{eq:PIm_nc} with $m=2$:
\begin{align}
    \label{eq:PI2_nc}
    \tag*{$\text{PI}_2^{\text{nc}}$}
    u^{(4)}
    + 5 (u')^2 
    + 10 u^3
    + 5 u'' \, u
    + 5 u \, u''
    + z
    &= 0
    .
\end{align}

\begin{rem}
In the matrix case, it coincides with eq. (3.4) from the paper \cite{Pic} after a shift of $u$ and $z$. 
\end{rem}

\subsection{The third member}
\label{sec:dPI3_nc}
For the third member, we have
\begin{gather}
    \begin{aligned}
    B_{7, n}
    &= c_7
    ,
    &&&
    B_{5, n}
    &= c_5 - d_6 \alpha_n,
    &&&
    B_{3, n}
    &= c_3 - d_4 \alpha_n - d_6 \br{
    \alpha_{n - 1} \alpha_n
    + \alpha_n^2
    + \alpha_n \alpha_{n + 1}
    },
    \end{aligned}
    \\[2mm]
    \begin{aligned}
    B_{1, n}
    = c_1
    &- \,d_2 \alpha_n
    - d_4 \br{
    \alpha_{n - 1} \alpha_n
    + \alpha_n^2
    + \alpha_{n + 1} \alpha_n
    }
    - d_6 \left(
    \alpha _{n-2} \alpha _{n-1} \alpha _n
    + \alpha _{n-1}^2 \alpha _n
    + \alpha _n \alpha _{n-1} \alpha _n
    \right.
    \\
    &+ \, \alpha _{n-1} \alpha _n^2
    \left.
    + \alpha _n^3
    + \alpha _{n-1} \alpha _n \alpha _{n+1}
    + \alpha _n \alpha _{n+1} \alpha _n
    + \alpha _n^2 \alpha _{n+1}
    + \alpha _n \alpha _{n+1}^2
    + \alpha _n \alpha _{n+1} \alpha _{n+2}
    \right),
    \end{aligned}
    \\[2mm]
    \begin{aligned}
    P_{6, n}
    &= d_6,
    &&&&&
    P_{4, n}
    &= d_4 + d_6 \br{
    \alpha_{n - 1} + \alpha_n
    },
    \end{aligned}
    \\[2mm]
    P_{2, n}
    = d_2 + d_4 \br{
    \alpha_{n - 1} + \alpha_n
    }
    + d_6 \br{
    \alpha_{n - 2} \alpha_{n - 1} 
    + \alpha_{n - 1} \alpha_n
    + \alpha_{n} \alpha_{n - 1}
    + \alpha_{n} \alpha_{n + 1} 
    + \alpha_{n - 1}^2 + \alpha_n^2
    },
    \\[2mm]
    \begin{aligned}
    P_{0, n}
    = d_0 
    + d_2 \br{
    \alpha_{n - 1} + \alpha_n
    }
    + d_4 \br{
    \alpha_{n - 2} \alpha_{n - 1} 
    + \alpha_{n - 1} \alpha_n
    + \alpha_{n} \alpha_{n - 1}
    + \alpha_{n} \alpha_{n + 1} 
    + \alpha_{n - 1}^2 + \alpha_n^2
    }
    + d_6 \left(
    \alpha _{n-1}^3
    + \alpha _n^3
    \right.
    \\
    \left.
    + \, \alpha _{n-2} \alpha _{n-1}^2
    + \alpha _{n-2}^2 \alpha _{n-1}
    + \alpha _{n-1} \alpha _n^2
    + \alpha _{n-1}^2 \alpha _n
    + \alpha _n \alpha _{n-1}^2
    + \alpha _n \alpha _{n+1}^2
    + \alpha _n^2 \alpha _{n-1}
    + \alpha _n^2 \alpha _{n+1}
    \right.
    \\
    \left.
    + \, \alpha _{n-1} \alpha _{n-2} \alpha _{n-1}
    + \alpha _{n-1} \alpha _n \alpha _{n-1}
    + \alpha _{n-1} \alpha _n \alpha _{n+1}
    + \alpha _n \alpha _{n-2} \alpha _{n-1}
    + \alpha _n \alpha _{n-1} \alpha _n
    \right.
    \\
    \left.
    + \alpha _n \alpha _{n+1} \alpha _{n-1}
    + \alpha _n \alpha _{n+1} \alpha _n
    + \alpha _n \alpha _{n+1} \alpha _{n+2}
    \right)
    .
    \end{aligned}
\end{gather}
Therefore, Theorem \ref{thm:dPIm_nc} gives us
\begin{align}
    \label{eq:dPI3_nc}
    \tag*{$\text{d-PI}_3^{\text{nc}}$}
    c_0 - n
    =
    d_0 \alpha_n
    + d_2 \br{
    \alpha_{n - 1} \alpha_n 
    + \alpha_n^2
    + \alpha_n \alpha_{n + 1}
    }
    + d_4 \Big(
    \alpha _{n-2} \alpha _{n-1} \alpha _n
    + \alpha _{n-1} \alpha _n \alpha _{n+1}
    + \alpha _n \alpha _{n+1} \alpha _{n+2}
    \Big.
    \\
    \left.
    + \, \br{
    \alpha_{n - 1} + \alpha_n
    }^2 \alpha_n
    - \alpha_n^3
    + \alpha_n \br{
    \alpha_{n + 1} + \alpha_n
    }^2
    \right)
    + d_6 \Big(
    \alpha _{n-3} \alpha _{n-2} \alpha _{n-1} \alpha _n
    + \alpha _{n-2} \alpha _{n-1} \alpha _n \alpha _{n+1}
    \Big.
    \\
    + \, \alpha _{n-1} \alpha _n \alpha _{n+1} \alpha _{n+2}
    + \alpha _n \alpha _{n+1} \alpha _{n+2} \alpha _{n+3}
    + \alpha _{n-1} \alpha _n \alpha _{n+1} \alpha _n
    + \alpha _{n-1} \alpha _n \alpha _{n+1} \alpha _{n+2}
    \\
    + \, \alpha _n \alpha _{n-2} \alpha _{n-1} \alpha _n
    + \alpha _n \alpha _{n-1} \alpha _n \alpha _{n+1}
    + \alpha _n \alpha _{n+1} \alpha _{n-1} \alpha _n
    + \alpha _{n-2} \alpha _{n-1} \alpha _n^2
    + \alpha _{n-2} \alpha _{n-1}^2 \alpha _n
    \\
    + \, \alpha _{n-2}^2 \alpha _{n-1} \alpha _n
    + \alpha _{n-1} \alpha _n \alpha _{n+1}^2
    + \alpha _{n-1} \alpha _n^2 \alpha _{n+1}
    + \alpha _{n-1}^2 \alpha _n \alpha _{n+1}
    + \alpha _n \alpha _{n+1} \alpha _{n+2}^2
    \\
    + \, \alpha _n \alpha _{n+1}^2 \alpha _{n+2}
    + \alpha _n^2 \alpha _{n+1} \alpha _{n+2}
    + \alpha _{n-1} \alpha _{n-2} \alpha _{n-1} \alpha _n
    \\
    \left.
    + \, \br{
    \alpha_{n - 1} + \alpha_n
    }^3 \alpha_n
    - \alpha_n^4
    + \alpha_n \br{
    \alpha_{n + 1} + \alpha_n
    }^3
    \right)
    .
\end{align}

\subsubsection{Continuous limit}
Taking the following change of variables 
\begin{align}
    \alpha_n
    &= 1 + \varepsilon^2 \, u
    ,
    &
    z
    &= n \, \varepsilon,
    &
    d_6
    &= \varepsilon^{-8}
    &
    d_4
    &= -14 \varepsilon^{-8}
    &
    d_2
    &= 70 \varepsilon^{-8},
    &
    d_0
    &= - 140 \varepsilon^{-8},
    &
    c_0 + n
    &= z + 35 \varepsilon^{-8}
\end{align}
and the subsequent limit $\varepsilon \to 0$, the \ref{eq:dPI3_nc} maps to the \ref{eq:PIm_nc} with $m=3$:
\begin{align}
    \label{eq:PI3_nc}
    \tag*{$\text{PI}_3^{\text{nc}}$}
    u^{(6)}
    + 7 u u^{(4)}
    + 7 u^{(4)} u
    &+ \, 21 \br{u''}^2
    + 21 u^2 u''
    + 21 u'' u^2
    + 28 u u'' u
    \\[1.1mm]
    &+ \, 28 u \br{u'}^2
    + 28 \br{u'}^2 u
    + 14 u' u u'
    + 14 u' u^{(3)}
    + 14 u^{(3)} u'
    + 35 u^4
    + z
    = 0
    .
\end{align}
Under the commutative reduction, it recovers the third member of the first Painlevé hierarchy
\begin{align}
    \label{eq:PI3}
    \tag*{$\text{PI}_3$}
    u^{(6)}
    + 14 u^{(4)} u
    + 70 u^2 u''
    + 21 \br{u''}^2
    + 70 u \br{u'}^2
    + 28 u^{(3)} u'
    + 35 u^4
    + z
    &= 0
    .
\end{align}
\begin{rem}
Note that \ref{eq:PI3_nc} in the matrix case gives the third member of the matrix first Painlevé hierarchy \cite{Pic} (up to a scaling and a shift). 
\end{rem}

\section{Stationary flows of the Volterra lattice hierarchy}
\label{sec:voltlat}

The Volterra lattice is the special case $p = 1$ of the Narita--Itoh--Bogoyavlensky lattice
\begin{align}
    \label{eq:NIB_nc}
    &&&&
    u_t
    &= \sum_{k = 1}^p \br{
    u_{n - k} u_n - u_n u_{n + k}
    }
    ,
    &
    p\in\mathbb{N}
    .
    &&&&
\end{align}
Due to the existence of the recursion operator constructed in \cite{casati2020recursion}, one obtains the Narita--Itoh--Bogoyavlensky hierarchy and, in particular, the Volterra lattice hierarchy. 
The latter has the form
\begin{align}
    \label{eq:VLr_nc}
    \tag*{$\text{VL}_r^{\text{nc}}$}
    \partial_{t_r} u_n
    &= S_{r}^r \br{n + r - 2} \, u_n
    - u_n \, S_{r}^r \br{n + r}
    ,
\end{align}
where $S_s^k(n)$ denotes the non-commutative Svinin polynomial defined in \eqref{eq:spoly_nc}.
\begin{rem}
The original Svinin polynomials are related to $S_s^k(n)$ by the change of indices $n \mapsto -n$, which yields
\begin{align}
    &&
    S_{0}^k \br{n}
    &= 1,
    &
    S_{s}^k \br{n}
    &= \sum_{j = 0}^{s - 1}
    u_{n + j + k - 1} S_{s - j}^{k - 1} \br{n + j}
    .
    &&
\end{align}
\end{rem}

Taking the stationary flows of the non-commutative Volterra lattice hierarchy \ref{eq:VLr_nc}, we arrive at the non-commutative first d-Painlevé hierarchy \ref{eq:dPIm_nc} containing only an alternating contribution. 
\begin{prop}
The stationary flows of the Volterra lattice hierarchy {\rm{\ref{eq:VLr_nc}}} give rise to the particular case of the {\rm{\ref{eq:dPIm_nc}}} hierarchy with $c_n = c_1 (-1)^n + c_0$.
\end{prop}

\begin{proof}
By Lemma \ref{thm:bdcond_nc}, we have
\begin{align}
    \alpha_n \, S_r^r (n + r - 1)
    - S_r^r (n + r - 1) \, \alpha_n
    &= - \alpha_n \, D_{r, n} 
    + D_{r, n - 1} \, \alpha_n
    ,
\end{align}
where, as before, $D_{r, n} = S_{r + 1}^r(n + r) - S_r^r(n + r) - S_{r}^r (n + r - 1)$. We define 
\begin{align}
    F_{r, n} 
    := S_{r}^r(n + r) 
    - c_1(-1)^n - c_0
\end{align}
and rewrite the right-hand side of the \ref{eq:VLr_nc} as follows
\begin{multline}
    u_n \, \br{
    F_{r, n}
    - F_{r, n - 1}
    } 
    + u_n \, S_{r}^r (n + r - 1)
    - S_r^r (n + r - 1) \, u_n
    + \br{
    F_{r, n - 1} - F_{r, n - 2}
    } \, u_n
    \\[2mm]
    = u_n \, \br{
    F_{r, n}
    - F_{r, n - 1}
    - D_{r, n}
    }
    + \br{
    F_{r, n - 1} - F_{r, n - 2}
    + D_{r, n - 1}
    } \, u_n
    .
\end{multline}
Taking the stationary flows $\alpha_n := u_n\big|_{t_r = \const}$ and due to the arbitrariness of $\alpha_n$ together with the assumption $\alpha_n \neq 0$, the coefficients of $\alpha_n$ must equal to zero, yielding the system
\begin{align}
    &&
    F_{r, n}
    - F_{r, n - 1}
    - D_{r, n}
    &= 0
    ,
    &
    F_{r, n}
    - F_{r, n - 1}
    + D_{r, n}
    &= 0.
    &&
\end{align}
Eliminating $D_{r, n}$ from the system, we obtain the equation
\begin{align}
    F_{r, n}
    = F_{r, n - 1}
    ,
\end{align}
whose general solution is $F_{r, n} = \const$. 
Finally, summing over $r = 1, \dots, m+1$ with arbitrary coefficients $\tilde d_r$ and redefining the parameters, we obtain the hierarchy \ref{eq:dPIm_nc} with $c_n = c_1(-1)^n + c_0$.
\end{proof}

\begin{rem}
The stationary reduction of the \ref{eq:VLr_nc} hierarchy also implies $S_r^{r + 1}(n + r) = \const$ due to the identity
\begin{align}
    \alpha_n \, S_r^r (n + r)
    - S_{r}^r(n + r - 2) \, \alpha_n
    &= S_{r}^{r + 1} (n + r)
    - S_r^{r + 1} (n + r - 1)
    .
\end{align}

\end{rem}

\begin{rem}
A similar commutative statement was considered in \cite{gordoa2005non}, where the first several members of the discrete first Painlevé hierarchy were derived from a non-isospectral extension of the Volterra hierarchy. 
\end{rem}

\begin{rem}
There exist other non-commutative versions of the Volterra lattice, whose reductions were studied in \cite{adler2020}. 
In particular, the author derived two d-PI equations different from \ref{eq:dPI1_nc} (see~\eqref{eq:dPI11_m}~and~\eqref{eq:dPI12_m}) and showed that they are consistent with the corresponding Volterra lattice. 
\end{rem}

\begin{rem}
Similarly to \cite{adler2020}, one can verify that the first member of \ref{eq:dPIm_nc} is consistent with the first flow of the Volterra hierarchy \ref{eq:VLr_nc}. Indeed, consider the first member \ref{eq:dPI1_nc} and set
\begin{align}
    \label{eq:F1}
    F_{1, n}
    &:= \br{
    u_{n - 1} u_n + u_n^2 + u_n u_{n + 1}
    + t_1 u_n + \br{n - c_0}
    }
    .
\end{align}
The first flow of the Volterra hierarchy is
\begin{align}
    \label{eq:VL1_nc}
    \tag*{$\text{VL}_1^{\text{nc}}$}
    \partial_{t_1} u_n
    &= u_{n - 1} u_n - u_n u_{n + 1}
    .
\end{align}
Differentiating \eqref{eq:F1} using \ref{eq:VL1_nc}, we obtain an expression which can be rewritten as follows
\begin{align}
    \partial_{t_1} \br{F_{1, n}}
    &=\br{F_{1, n} - F_{1, n - 1}} u_n
    + u_n \br{F_{1, n + 1} - F_{1, n}}
    .
\end{align}
Hence, the constraint $F_{1,n}=0$ is preserved by the first Volterra flow.
\end{rem}

\section{Discussions}
\label{sec:disc}

As demonstrated by the explicit examples, our \ref{eq:dPIm_nc} hierarchy appears to be related to the \ref{eq:PIm_nc} hierarchy. To make this connection precise, the system \ref{eq:dPIm_nc} should be reformulated in terms of a discrete version of the recursion operator, analogous to \eqref{eq:Lenard_nc}. Such an operator could, in principle, be derived from the discrete Korteweg-de Vries hierarchy. We note that this problem remains open even in the commutative setting.

\medskip

Another non-commutative version of the d-PI equation, given by \eqref{eq:dPI10_nc}, was obtained in~\cite{bobrova2024affine} through a composition of Bäcklund transformations associated with the non-commutative fourth Painlevé equation. Since Painlevé IV admits a matrix hierarchy~\cite{gordoa2021matrix} together with corresponding Bäcklund transformations, one may expect a different hierarchy of the first discrete Painlevé equation—with \eqref{eq:dPI10_nc} as its first member—to arise from an appropriate composition of the transformations constructed in~\cite{gordoa2021matrix}.

\medskip

It is also natural to investigate analogous hierarchies for other non-commutative discrete Painlevé equations, for example those introduced in~\cite{doliwa2013non}, \cite{bobrova2024affine}, \cite{bobrova2025non}, as well as to study their connection with isomonodromic representations, orthogonal polynomials, and integrable systems.

\medskip

In particular, a d-PII hierarchy can be derived in the same spirit as in~\cite{cresswell1999discrete}. To illustrate this idea, consider the first member of the hierarchy, whose commutative counterpart was obtained in~\cite{joshi1992nonlinear}. The~linear problem reads as
\begin{align}
    \label{eq:lsysmat_nc}
    \left\{
    \begin{array}{lcl}
        \Phi_{n + 1}(\lambda)
        &=& L_n(\lambda)\, \Phi_n(\lambda), 
        \\[2mm]
        \partial_\lambda \Phi_n(\lambda)
        &=& M_n(\lambda)\, \Phi_n(\lambda),
    \end{array}
    \right.
\end{align}
with the matrices $L_n\br{\lambda} = L_n$, $M_n\br{\lambda} = M_n \in \Mat_2\br{\mathcal{R}}$ given by
\begin{align}
    &&
    L_n(\lambda)
    &= 
    \begin{pmatrix}
        \lambda & \alpha_n \\[1mm]
        \alpha_n & \lambda^{-1}
    \end{pmatrix},
    &
    M_n(\lambda)
    &= 
    \begin{pmatrix}
        A_n & B_n \\[1mm]
        C_n & -A_n
    \end{pmatrix}.
    &&
\end{align}
Here $\Phi_n\br{\lambda}\in \Mat_2\br{\mathcal{R}}$ and matrix entries $A_n$, $B_n$, $C_n$ belong to $\mathcal{R}$. 

\begin{rem}
The scalar system \eqref{eq:lsys_nc} can also be rewritten in the matrix form \eqref{eq:lsysmat_nc}. In this case, the matrices $L_n$ and $M_n$ are
\begin{align}
    &&
    L_n \br{\lambda}
    &= 
    \begin{pmatrix}
        \lambda \alpha_{n + 1}^{-1} & - \alpha_{n + 1}^{-1}
        \\[1mm]
        1 & 0
    \end{pmatrix}
    ,
    &
    &
    M_n(\lambda)
    &= 
    \begin{pmatrix}
        \lambda p_{n + 1} + b_{n + 1} & - p_{n+1}\alpha_{n + 1} \alpha_n^{-1} \\[1mm]
        p_n \alpha_n & b_n
    \end{pmatrix}
    ,
    &&
\end{align}
where, as before, $p_n := a_n \alpha_n^{-1}$.
\end{rem}
The~compatibility condition of \eqref{eq:lsysmat_nc}
\begin{align}
    \label{eq:ccond}
    \partial_\lambda L_n
    = M_{n + 1} L_n - L_n M_n
\end{align}
imposes constraints on $A_n$, $B_n$, and $C_n$. Expanding them into Laurent series in $\lambda$ beginning at $\lambda^{-3}$, one finds
\begin{gather}
    A_n
    = \tfrac12 d_1 \lambda^{-3} + \bigl(
    - d_1 \alpha_n \alpha_{n-1} 
    + n 
    + \tfrac12(d_{-1} - 1)
    \bigr)\lambda^{-1} + \tfrac12 d_1 \lambda,
    \\[2mm]
    \begin{aligned}
        B_n &= - d_1 \alpha_{n-1} \lambda^{-2} + d_1 \alpha_n, 
        &&&&&&&
        C_n &= - d_1 \alpha_n \lambda^{-2} + d_1 \alpha_{n-1}
        .
    \end{aligned}
\end{gather}
As a result, \eqref{eq:ccond} leads to the non-commutative discrete Painlevé II equation:
\begin{align}
    \label{eq:dPII1_nc}
    \tag*{$\mathrm{d\text{-}PII}^{\mathrm{nc}}_1$}
    d_1\bigl(
        \alpha_{n+1}(1 - \alpha_n^2) + (1 - \alpha_n^2)\alpha_{n-1}
    \bigr) + \alpha_n (2n + d_{-1}) = 0,
\end{align}
which, under the continuous limit
\begin{align}
    \alpha_n &= - \varepsilon\, u, &
    z &= \varepsilon n, &
    d_1 &= \varepsilon^{-3}, &
    d_{-1} + 2n &= - 2\varepsilon^{-3} - z\, \varepsilon^{-1},
\end{align}
reduces to the $\mathrm{PII}^0$ equation with zero free parameter~\cite{Adler_Sokolov_2020_1}:
\begin{align}
    u'' = 2 u^3 + z u.
\end{align}
Higher members of the hierarchy can be derived by expanding $A_n$, $B_n$, and $C_n$ into Laurent series in $\lambda$ beginning at $\lambda^{-(2m+1)}$, where $m$ labels the $m$-th member of the hierarchy. It would be interesting to identify another family of polynomials naturally associated with this hierarchy.

\medskip
Finally, we mention that a recent work~\cite{svinin2025volterra} established a connection between Svinin polynomials and Somos-$N$ sequences. A natural question is whether a similar correspondence also exists between non-commutative Svinin polynomials and the non-commutative Somos-$N$ equations introduced in~\cite{bobrova2023non}.

    \appendix

\section{Proof of Proposition \ref{thm:Lenard_nc} and Lemma \ref{thm:Gn_rec}}
\label{app:Lenard_nc}
Below we will omit the argument of the Lenard operator and the upper-script $^\text{nc}$. Before we start, let us list useful identities:
\begin{itemize}
\item formal integration by parts
\begin{align}
    \label{eq:inbp}
    &f \, g
    = \partial^{-1} \br{f' \, g}
    + \partial^{-1} \br{f \, g'}
    ,
\end{align}
\item symmetry and skew-symmetry properties:
\begin{align}
    \label{eq:acs}
    &\lbr{f, g}_+
    = \lbr{g, f}_+
    ,
    \\[2mm]
    \label{eq:cs}
    &\lbr{f, g}
    = - \lbr{g, f}
    ,
\end{align}
\item Leibniz rules:
\begin{align}
    \label{eq:acl}
    &\br{\lbr{f, g}_+}'
    = \lbr{f', g}_+ + \lbr{f, g'}_{+}
    ,
    \\[2mm]
    \label{eq:cl}
    &\br{\lbr{f, g}}'
    = \lbr{f', g} + \lbr{f, g'}
    ,
\end{align}
\item relations between anticommutators and commutators:
\begin{align}
    \label{eq:ac_c}
    &\lbr{f,\lbr{g,h}}
    = \lbr{\lbr{f,g}_+, h}_+
    - \lbr{g,\lbr{f,h}_+}_+
    ,
    \\[2mm]
    \label{eq:ac_3}
    &\lbr{f, \lbr{g, h}_+}
    = \lbr{\lbr{f, g}, h}_+
    + \lbr{g, \lbr{f, h}}_+
    ,
\end{align}
\item Jacobi identity:
\begin{align}
    \lbr{f, \lbr{g,h}}
    = \lbr{\lbr{f,g}, h}
    + \lbr{g, \lbr{f,h}}
    .
\end{align}
\end{itemize}
Note that we use the boundary condition in \eqref{eq:Lenard_nc} while formally integrating by parts. 

\begin{proof}[Proof of Proposition \ref{thm:Lenard_nc}]
The proof is similar to the commutative case and is based on a formal calculus involving $\partial^{-1}$, where $\partial^{-1} \partial = \partial \partial^{-1} = 1$. We use the notation $\partial \equiv \prime$. 
Recall that $u = \cLen{1}$. 

\medskip
\textbullet \,\, We start with \eqref{eq:Lenard_nc} which can be rewritten as follows:
\begin{align}
    \cLen{n + 1}'
    &= \cLen{n}'''
    + 2 \lbr{\cLen{1}, \cLen{n}'}_+
    + \lbr{\cLen{1}', \cLen{n}}_+
    + G_n
    = \br{
    \cLen{n}'' 
    + \lbr{\cLen{1}, \cLen{n}}_+
    }'
    + \lbr{\cLen{1}, \cLen{n}'}_+
    + G_n
    ,
    \\[2mm]
    \label{eq:Gn}
    G_n
    &:= \lbr{\cLen{1}, \mathcal{G}_n}
    ,
    \qquad
    \mathcal{G}_n
    := \partial^{-1} \br{\lbr{\cLen{1}, \cLen{n}}}
    ,
\end{align}
or, after formal integrating, 
\begin{align}
    \label{eq:prefin}
    \cLen{n + 1}
    &= \cLen{n}''
    + \lbr{\cLen{1}, \cLen{n}}_+
    + \partial^{-1} \br{
    \lbr{\cLen{1}, \cLen{n}'}_+
    }
    + \partial^{-1} \br{G_n}
    .
\end{align}
In addition, the formula \eqref{eq:inbp} gives us
\begin{align}
    \label{eq:len_np1}
    \cLen{n + 1}
    &= \cLen{n}''
    + 2 \lbr{\cLen{1}, \cLen{n}}_+
    - \partial^{-1} \br{
    \lbr{\cLen{1}', \cLen{n}}_+
    }
    + \partial^{-1} \br{G_n}
    ,
\end{align}
thus
\begin{align}
    \label{eq:intform}
    \partial^{-1} \br{
    \lbr{\cLen{1}', \cLen{n}}_+
    }
    &= \cLen{n}''
    + 2 \lbr{\cLen{1}, \cLen{n}}_+
    - \cLen{n + 1}
    + \partial^{-1} \br{G_n}
    .
\end{align}

\medskip
\textbullet \,\, 
Now we are aiming to compute $\partial^{-1}\br{\lbr{\cLen{1}, \cLen{n}'}_+}$ in \eqref{eq:prefin} by using the recursion formula and \eqref{eq:intform}. Let us look at $\partial^{-1}\br{\lbr{\cLen{l}, \cLen{n}'}_+}$:
\begin{align}
    \partial^{-1}\br{\lbr{\cLen{l}, \cLen{n}'}_+}
    \overset{\eqref{eq:Lenard_nc}}{=} 
    {\partial^{-1} \br{\lbr{
    \cLen{l}, \cLen{n - 1}'''
    }_+}}
    &+ \, 2 \partial^{-1} 
    \br{\lbr{
    \cLen{l}, \lbr{\cLen{1}, \cLen{n - 1}'}_+
    }_+}
    \\
    &+ \, \partial^{-1} \br{\lbr{
    \cLen{l}, \lbr{\cLen{1}', \cLen{n - 1}}_+
    }_+}
    + \partial^{-1} \br{\lbr{
    \cLen{l}, G_{n - 1}
    }_+}
    .
\end{align}
The first term can be integrated formally by parts
\begin{align}
    {\partial^{-1} \br{\lbr{
    \cLen{l}, \cLen{n - 1}'''
    }_+}}
    &\overset{\eqref{eq:inbp}}{=}
    \lbr{
    \cLen{l}, \cLen{n - 1}''
    }_+
    - \partial^{-1} \br{\lbr{
    \cLen{l}', \cLen{n - 1}''
    }_+}
    \\
    &\overset{\eqref{eq:inbp}}{=}
    \lbr{
    \cLen{l}, \cLen{n - 1}''
    }_+
    - \lbr{
    \cLen{l}', \cLen{n - 1}'
    }_+
    + \partial^{-1} \br{\lbr{
    \cLen{l}'', \cLen{n - 1}'
    }_+}
    .
\end{align}
In order to apply \eqref{eq:intform} for the third term, we rewrite it with the help of \eqref{eq:ac_c}:
\begin{align}
    \partial^{-1} \br{\lbr{
    \cLen{l}, \lbr{\cLen{1}', \cLen{n - 1}}_+
    }_+}
    &\overset{\eqref{eq:ac_c}}{=}
    \partial^{-1} \br{\lbr{
    \cLen{n - 1}, \lbr{\cLen{1}', \cLen{l}}_+
    }_+}
    + \partial^{-1} \br{
    \lbr{
    \cLen{1}', \lbr{
    \cLen{n - 1}, \cLen{l}
    }
    }}
    \\[2mm]
    &\overset{\phantom{\eqref{eq:ac_c}}}{=} 
    \partial^{-1} \br{\lbr{
    \cLen{n - 1}, \partial \br{\partial^{-1} \lbr{\cLen{1}', \cLen{l}}_+}
    }_+}
    + \partial^{-1} \br{
    \lbr{
    \cLen{1}', \lbr{
    \cLen{n - 1}, \cLen{l}
    }
    }}
    \\[2mm]
    &\overset{\phantom{\eqref{eq:ac_c}}}{=} 
    \lbr{
    \cLen{n - 1}, \partial^{-1} \br{\lbr{\cLen{1}', \cLen{l}}_+}
    }_+
    - \partial^{-1} \br{
    \lbr{
    \cLen{n - 1}', \partial^{-1} \br{\lbr{\cLen{1}', \cLen{l}}_+}
    }_+}
    + \partial^{-1} \br{
    \lbr{
    \cLen{1}', \lbr{
    \cLen{n - 1}, \cLen{l}
    }
    }}
    \\[2mm]
    &\overset{\eqref{eq:intform}}{=}
    \lbr{
    \cLen{n - 1}, \cLen{l}''
    + 2 \lbr{\cLen{1}, \cLen{l}}_+
    - \cLen{l + 1}
    + \partial^{-1} \br{G_l}
    }_+
    \\
    &\qquad
    - \partial^{-1} \br{
    \lbr{
    \cLen{n - 1}', \cLen{l}''
    + 2 \lbr{\cLen{1}, \cLen{l}}_+
    - \cLen{l + 1}
    + \partial^{-1} \br{G_l}
    }_+}
    + \partial^{-1} \br{
    \lbr{
    \cLen{1}', \lbr{
    \cLen{n - 1}, \cLen{l}
    }
    }}
    \\[2mm]
    &\overset{\phantom{\eqref{eq:ac_c}}}{=} 
    \lbr{
    \cLen{n - 1}, \cLen{l}''
    }_+
    + 2 
    \lbr{
    \cLen{n - 1}, \lbr{\cLen{1}, \cLen{l}}_+
    }_+
    - \lbr{\cLen{n - 1}, \cLen{l + 1}}_+
    + \lbr{\cLen{n - 1}, \partial^{-1} \br{G_l}}_+
    \\
    &\qquad
    - \partial^{-1} \br{\lbr{
    \cLen{n - 1}', \cLen{l}''
    }_+}
    - 2 \, \partial^{-1} \br{
    \lbr{
    \cLen{n - 1}', \lbr{\cLen{1}, \cLen{l}}_+
    }_+}
    + \partial^{-1} \br{\lbr{\cLen{n - 1}', \cLen{l + 1}}_+}
    \\
    &\qquad
    - \partial^{-1} \br{\lbr{\cLen{n - 1}', \partial^{-1} G_l}_+}
    + \partial^{-1} \br{
    \lbr{
    \cLen{1}', \lbr{
    \cLen{n - 1}, \cLen{l}
    }
    }}
    .
\end{align}
Note also that, thanks to \eqref{eq:ac_c}, the second term takes the form
\begin{align}
    \partial^{-1} \br{
    \lbr{
    \cLen{l}, \lbr{\cLen{1}, \cLen{n - 1}'}_+
    }_+}
    &= \partial^{-1} \br{\lbr{
    \cLen{n - 1}', \lbr{\cLen{l}, \cLen{1}}_+
    }_+}
    + \partial^{-1} \br{\lbr{
    \cLen{1}, \lbr{\cLen{n - 1}', \cLen{l}}
    }}
    .
\end{align}
Thus, all together
\begin{align}
    \partial^{-1} \br{\lbr{\cLen{l}, \cLen{n}'}_+}
    \overset{\phantom{\eqref{eq:ac_c}}}{=}  
    \lbr{
    \cLen{l}, \cLen{n - 1}''
    }_+
    &- \, \lbr{
    \cLen{l}', \cLen{n - 1}'
    }_+
    + \orange{\partial^{-1} \br{\lbr{
    \cLen{l}'', \cLen{n - 1}'
    }_+}}
    \\
    &
    + \, 2 \, \tomato{\partial^{-1} \br{\lbr{
    \cLen{n - 1}', \lbr{\cLen{l}, \cLen{1}}_+
    }_+}}
    + 2 \partial^{-1} \br{\lbr{
    \cLen{1}, \lbr{\cLen{n - 1}', \cLen{l}}
    }}
    \\
    &+ \, \lbr{
    \cLen{n - 1}, \cLen{l}''
    }_+
    + 2 
    \lbr{
    \cLen{n - 1}, \lbr{\cLen{1}, \cLen{l}}_+
    }_+
    - \lbr{\cLen{n - 1}, \cLen{l + 1}}_+
    + \lbr{\cLen{n - 1}, \partial^{-1} \br{G_l}}_+
    \\
    &
    - \, \orange{\partial^{-1}\br{\lbr{
    \cLen{n - 1}', \cLen{l}''
    }_+}}
    - 2 \, \tomato{\partial^{-1}\br{
    \lbr{
    \cLen{n - 1}', \lbr{\cLen{1}, \cLen{l}}_+
    }_+}}
    + \partial^{-1} \br{\lbr{\cLen{n - 1}', \cLen{l + 1}}_+}
    \\
    &
    - \, \partial^{-1} \br{\lbr{\cLen{n - 1}', \partial^{-1} \br{G_l}}_+}
    + \partial^{-1} \br{
    \lbr{
    \cLen{1}', \lbr{
    \cLen{n - 1}, \cLen{l}
    }
    }}
    + \partial^{-1} \br{\lbr{
    \cLen{l}, G_{n - 1}
    }_+}
    \\[2mm]
    \overset{\eqref{eq:acs}}{=}
    \lbr{
    \cLen{l}, \cLen{n - 1}''
    }_+
    &- \, \lbr{
    \cLen{l}', \cLen{n - 1}'
    }_+
    + \lbr{
    \cLen{n - 1}, \cLen{l}''
    }_+
    \\
    &+ \, 2 
    \lbr{
    \cLen{n - 1}, \lbr{\cLen{1}, \cLen{l}}_+
    }_+
    - \lbr{\cLen{n - 1}, \cLen{l + 1}}_+
    + \partial^{-1} \br{\lbr{\cLen{n - 1}', \cLen{l + 1}}_+}
    \\
    &+ \, 2 \, \mseagreen{\partial^{-1} \br{\lbr{
    \cLen{1}, \lbr{\cLen{n - 1}', \cLen{l}}
    }}}
    + \mseagreen{\partial^{-1} \br{
    \lbr{
    \cLen{1}', \lbr{
    \cLen{n - 1}, \cLen{l}
    }}
    }}
    \\
    &
    + \, \flowerblue{\lbr{\cLen{n - 1}, \partial^{-1} \br{G_l}}_+}
    - \flowerblue{\partial^{-1} \br{
    \lbr{\cLen{n - 1}', \partial^{-1} \br{G_l}}_+}}
    + \partial^{-1} \br{\lbr{
    \cLen{l}, G_{n - 1}
    }_+}
    \\[2mm]
    \overset{\eqref{eq:inbp}}{=}
    \lbr{
    \cLen{l}, \cLen{n - 1}''
    }_+
    &- \, \lbr{
    \cLen{l}', \cLen{n - 1}'
    }_+
    + \lbr{
    \cLen{n - 1}, \cLen{l}''
    }_+
    \\
    &+ \, 2 
    \lbr{
    \cLen{n - 1}, \lbr{\cLen{1}, \cLen{l}}_+
    }_+
    - \lbr{\cLen{n - 1}, \cLen{l + 1}}_+
    + \partial^{-1}\br{\lbr{\cLen{n - 1}', \cLen{l + 1}}_+}
    \\
    &+ \, \lbr{\cLen{1}, \lbr{\cLen{n-1}, \cLen{l}}}
    + \partial^{-1} \br{\lbr{
    \cLen{1}, \lbr{\cLen{n - 1}', \cLen{l}}
    - \lbr{\cLen{n - 1}, \cLen{l}'}
    }}
    \\
    &+ \, \partial^{-1} \br{\lbr{
    \cLen{n - 1}, G_{l}
    }_+}
    + \partial^{-1} \br{\lbr{
    \cLen{l}, G_{n - 1}
    }_+}
    .
\end{align}

\medskip
\textbullet \,\, Set $l = 1$ and apply the recurrence until $\partial^{-1} \br{\lbr{\cLen{n - 1}', \cLen{l + 1}}_+}$ is equal to $\partial^{-1} \br{\lbr{\cLen{1}', \cLen{n}}_+}$, which can be formally integrated by parts as follows: $\partial^{-1} \br{\lbr{\cLen{1}', \cLen{n}}_+} = \lbr{\cLen{1}, \cLen{n}}_+ - \partial^{-1} \br{\lbr{\cLen{1}, \cLen{n}'}_+}$. Thus,
\begin{align}
    2 
    &\partial^{-1} \lbr{\cLen{1}, \cLen{n}'}_+
    \\
    &= \, \lbr{\cLen{1}, \cLen{n}}_+ 
    + \sum_{j = 1}^{n - 1} \left(
    2 \lbr{\cLen{n - j}, \cLen{j}''}_+
    - \lbr{\cLen{j}', \cLen{n - j}'}_+
    + 2 \lbr{\cLen{n - j}, \lbr{\cLen{1}, \cLen{j}}_+}_+
    - \lbr{\cLen{n - j}, \cLen{j + 1}}_+
    \right)
    \\
    &\phantom{= \lbr{\cLen{1}, \cLen{n}}=}
    + \, \sum_{j = 1}^{n - 1} \br{
    \mseagreen{\lbr{\cLen{1}, \lbr{\cLen{n-j}, \cLen{j}}}}
    + \partial^{-1} \br{\lbr{
    \cLen{1}, \orange{\lbr{\cLen{n - j}', \cLen{j}}
    - \lbr{\cLen{n - j}, \cLen{j}'}}
    }
    }}
    \\
    &\phantom{= \lbr{\cLen{1}, \cLen{n}}=}
    + \, \partial^{-1} \sum_{j = 1}^{n - 1} \br{ 
    \tomato{\lbr{
    \cLen{n - j}, G_{j}
    }_+
    + \lbr{
    \cLen{j}, G_{n - j}
    }}
    }
    \\[2mm]
    &= \, \lbr{\cLen{1}, \cLen{n}}_+ 
    + \sum_{j = 1}^{n - 1} \left(
    2 \lbr{\cLen{n - j}, \cLen{j}''}_+
    - \lbr{\cLen{j}', \cLen{n - j}'}_+
    + 2 \lbr{\cLen{n - j}, \lbr{\cLen{1}, \cLen{j}}_+}_+
    - \lbr{\cLen{n - j}, \cLen{j + 1}}_+
    \right)
    \\
    &\phantom{= \lbr{\cLen{1}, \cLen{n}}=}
    + \, 2 \, \partial^{-1} \sum_{j = 1}^{n - 1} \br{
    \lbr{
    \cLen{1}, \lbr{\cLen{n - j}', \cLen{j}}
    }
    + \lbr{
    \cLen{n - j}, G_j
    }_+
    }
    .
\end{align}
After substituting this into \eqref{eq:prefin} and using the definition of $G_n$, we obtain:
\begin{align}
    \cLen{n + 1}
    = \cLen{n}''
    &+ \, \tfrac32 \lbr{\cLen{1}, \cLen{n}}_+
    + \, \tfrac12 \sum_{j = 1}^{n - 1} \left(
    \lbr{\cLen{n - j}, \, 2 \cLen{j}''
    + 2 \lbr{\cLen{1}, \cLen{j}}_+
    - \cLen{j + 1}
    }_+
    - \lbr{\cLen{j}', \cLen{n - j}'}_+
    \right)
    \\
    & + \, \partial^{-1} \sum_{j = 1}^{n}
    \br{
     \lbr{
    \cLen{1}, \lbr{\cLen{n - j}', \cLen{j}}}
    + \lbr{
    \cLen{n - j}, G_j
    }_+
    }
    .
\end{align}
Thanks to Lemma \ref{thm:Gn_rec}, the last sum is equal to zero, and we are done.
\end{proof}

\begin{proof}[Proof of Lemma \ref{thm:Gn_rec}]
The proof proceeds by a direct computation, similarly to the proof of Proposition \ref{thm:Lenard_nc}, using the definition of $G_n$ together with equation \eqref{eq:prefin}. By recursively substituting the expressions for $\cLen{n}$ into \eqref{eq:Gn_rec}, one obtains an expression involving only $\cLen{1}$ and its derivatives.
To verify that \eqref{eq:Gn_rec} is indeed an identity, the key ingredient is the relation
\begin{align}
    \lbr{
    \cLen{1}^{(k)}, 
    \partial^{-1} \br{
    \lbr{\cLen{1}^{(l)}, \cLen{1}^{(m + 1)}}_{\pm}
    }
    }_{\pm}
    = \lbr{
    \cLen{1}^{(k)}, 
    \lbr{\cLen{1}^{(l)}, \cLen{1}^{(m)}}_{\pm}
    }_{\pm}
    - 
    \lbr{
    \cLen{1}^{(k)}, 
    \partial^{-1} \br{
    \lbr{\cLen{1}^{(l + 1)}, \cLen{1}^{(m)}}_{\pm}
    }
    }_{\pm}
    ,
\end{align}
where $_\pm$ denotes either the commutator or the anticommutator. In particular, setting $m=l$ yields
\begin{align}
    \lbr{
    \cLen{1}^{(k)}, 
    \partial^{-1} \br{
    \lbr{\cLen{1}^{(l)}, \cLen{1}^{(l + 1)}}_{+}
    }
    }_{\pm}
    = \frac12 \lbr{
    \cLen{1}^{(k)}, 
    \lbr{\cLen{1}^{(l)}, \cLen{1}^{(l)}}_{+}
    }_{\pm}
    ,
\end{align}
while for $m=l=k$ one obtains
\begin{align}
    \lbr{
    \cLen{1}^{(k)}, 
    \partial^{-1} \br{
    \lbr{\cLen{1}^{(k)}, \cLen{1}^{(k + 1)}}_{+}
    }
    }
    = 0
    .
\end{align}
Moreover, the property \eqref{eq:ac_3} with $f = g = \cLen{1}$ becomes 
\begin{align}
    \lbr{
    \cLen{1}, \lbr{
    \cLen{1}, h
    }_+
    } 
    &= \lbr{
    \cLen{1}, \lbr{
    \cLen{1}, h
    }
    }_+
    .
\end{align}
Repeated application of these identities systematically eliminates all nonlocal contributions appearing in~\eqref{eq:Gn_rec}. The remaining terms then cancel after straightforward, albeit lengthy, algebraic manipulations, thereby proving \eqref{eq:Gn_rec}.
\end{proof}
    
    \bibliographystyle{alpha}
    \bibliography{bib}

@article{Adler_Sokolov_2020_1,
  title={On matrix {P}ainlev{\'e} {II} equations},
  author={Adler, V. E. and Sokolov, V. V.},
  journal={Theoret. and Math. Phys.},
  volume={207},
  number={2},
  pages={\href{https://doi.org/10.4213/tmf10027}{188--201}},
  year={2021},
  note={\href{https://arxiv.org/abs/2012.05639}{arXiv:2012.05639}}
}

@article{Balandin_Sokolov_1998,
  title={On the {P}ainlev{\'e} test for non-{A}belian equations},
  author={Balandin, S. P. and Sokolov, V. V.},
  journal={Physics letters A},
  volume={246},
  number={3-4},
  pages={\href{https://www.sciencedirect.com/science/article/abs/pii/S0375960198003363?via\%3Dihub}{267--272}},
  year={1998},
  publisher={Elsevier}
}

@article{adler2020,
  title={Painlev{\'e} type reductions for the non-{A}belian {V}olterra lattices},
  author={Adler, V. E.},
  journal={Journal of Physics A: Mathematical and Theoretical},
  volume={54},
  number={3},
  pages={\href{https://doi.org/10.1088/1751-8121/abd21f}{035204}},
  year={2020},
  publisher={IOP Publishing},
  note={\href{https://arxiv.org/abs/2010.09021}{arXiv:2010.09021}}
}

@article{OS_1998a,
  title={Integrable evolution equations on associative algebras},
  author={Olver, P. J. and Sokolov, V. V.},
  journal={Communications in Mathematical Physics},
  volume={193},
  number={2},
  pages={\href{https://link.springer.com/article/10.1007/s002200050328}{245--268}},
  year={1998},
  publisher={Springer}
}

@article{Pic,
  title={On matrix {P}ainlev{\'e} hierarchies},
  author={Gordoa, P. R. and Pickering, A. and Zhu, Z. N.},
  journal={Journal of Differential Equations},
  volume={261},
  number={2},
  pages={\href{https://www.sciencedirect.com/science/article/pii/S0022039616300092}{1128--1175}},
  year={2016},
  publisher={Elsevier}
}

@article{mazzocco2007hamiltonian,
  title={The {H}amiltonian structure of the second {P}ainlev{\'e} hierarchy},
  author={Mazzocco, M. and Mo, M. Y.},
  journal={Nonlinearity},
  volume={20},
  number={12},
  pages={\href{https://doi.org/10.1088/0951-7715/20/12/006}{2845}},
  year={2007},
  note={\href{https://arxiv.org/abs/nlin/0610066}{arXiv:nlin/0610066}},
  publisher={IOP Publishing}
}

@article{bobrova2022classification,
  title={Classification of {H}amiltonian non-abelian {P}ainlev\'e type systems},
  author={Bobrova, I. and Sokolov, V.},
  journal={Journal of Nonlinear Mathematical Physics},
  volume = {30},
  pages={\href{https://doi.org/10.1007/s44198-022-00099-w}{646-662}},
  year={2023},
  note={\href{https://arxiv.org/abs/2209.00258}{arXiv:2209.00258}}
}

@article{manas2021matrix,
    title = {Matrix biorthogonal polynomials: {E}igenvalue problems and non-{A}belian discrete {P}ainlev\'e equations: a {R}iemann–{H}ilbert problem perspective},
    journal = {Journal of Mathematical Analysis and Applications},
    volume = {494},
    number = {2},
    pages = {\href{https://doi.org/10.1016/j.jmaa.2020.124605}{124605}},
    year = {2021},
    issn = {0022-247X},
    doi = {https://doi.org/10.1016/j.jmaa.2020.124605},
    url = {https://www.sciencedirect.com/science/article/pii/S0022247X2030768X},
    author = {Branquinho, A. and Moreno, A. and Mañas, M.}
}

@article{sakai2001rational,
  title={Rational Surfaces Associated with Affine Root Systems and Geometry of the {P}ainlev{\'e} Equations},
  author={Sakai, H.},
  journal={Communications in Mathematical Physics},
  volume={220},
  number={1},
  pages={\href{https://link.springer.com/article/10.1007/s002200100446}{165--229}},
  year={2001},
  publisher={Springer}
}

@book{conte2008painleve,
  title={The {P}ainlev{\'e} {H}andbook},
  author={Conte, R. and Musette, M.},
  year={\href{https://doi.org/10.1007/978-3-030-53340-3}{2020}},
  publisher={Springer Cham}
}

@article{olver2000classification,
  title={Classification of integrable one-component systems on associative algebras},
  author={Olver, P. J. and Wang, J. P.},
  journal={Proceedings of the London Mathematical Society},
  volume={81},
  number={3},
  pages={\href{https://doi.org/10.1112/S0024611500012582}{566--586}},
  year={2000},
  publisher={Cambridge University Press}
}

@article{bobrova2023non,
  title={Non-{A}belian discrete {T}oda chains and related lattices},
  author={Bobrova, I. and Retakh, V. and Rubtsov, V. and Sharygin, G.},
  journal={Physica D: Nonlinear Phenomena},
  year={2024},
  volume={464},
  pages={\href{https://doi.org/10.1016/j.physd.2024.134200}{134200}},
  note={\href{https://arxiv.org/abs/2311.11124}{arXiv:2311.11124}}
}

@article{kudryashov1997first,
  title={The first and second {P}ainlev{\'e} equations of higher order and some relations between them},
  author={Kudryashov, N. A.},
  journal={Physics Letters A},
  volume={224},
  number={6},
  pages={\href{https://www.sciencedirect.com/science/article/abs/pii/S0375960196007955}{353--360}},
  year={1997},
  publisher={Elsevier}
}

@article{cassatella2014singularity,
  title={Singularity confinement for matrix discrete {P}ainlev{\'e} equations},
  author={Cassatella-Contra, G. A and Manas, M. and Tempesta, P.},
  journal={Nonlinearity},
  volume={27},
  number={9},
  pages={\href{https://iopscience.iop.org/article/10.1088/0951-7715/27/9/2321}{2321}},
  year={2014},
  publisher={IOP Publishing}
}

@article{bobrova2024affine,
  title={Affine {W}eyl groups and non-{A}belian discrete systems: an application to the $d$-{P}ainlev\'e equations},
  author={Bobrova, I.},
  journal={Journal of Mathematical Physics, Analysis and Geometry (under review)},
  year={2024},
  note={\href{https://arxiv.org/abs/2403.18463}{arXiv:2403.18463}}
}

@article{doliwa2013non,
  title={Non-commutative $q$-{P}ainlev{\'e} {VI} equation},
  author={Doliwa, A.},
  journal={Journal of Physics A: Mathematical and Theoretical},
  volume={47},
  number={3},
  pages={\href{https://doi.org/10.1088/1751-8113/47/3/035203}{035203}},
  year={2013},
  publisher={IOP Publishing},
  note={\href{https://arxiv.org/abs/1310.6890}{arXiv:1310.6890}}
}

@article{bobrova2025non,
  title={On a non-commutative sixth $q$-{P}ainlev\'e system: from discrete system to surface theory},
  author={Bobrova, I.},
  journal={arXiv preprint \href{https://arxiv.org/abs/2507.22466}{arXiv:2507.22466}},
  year={2025}
}

@article{gordoa2021matrix,
  title={On matrix fourth {P}ainlev{\'e} hierarchies},
  author={Gordoa, P. R. and Pickering, A.},
  journal={Journal of Differential Equations},
  volume={271},
  pages={\href{https://doi.org/10.1016/j.jde.2020.08.013}{499--532}},
  year={2021},
  publisher={Elsevier}
}

@article{joshi1992nonlinear,
  title={Nonlinear nonautonomous discrete dynamical systems from a general discrete isomonodromy problem},
  author={Joshi, N. and Burtonclay, D. and Halburd, R. G.},
  journal={Letters in mathematical physics},
  volume={26},
  number={2},
  pages={\href{https://link.springer.com/article/10.1007/BF00398809}{123--131}},
  year={1992},
  publisher={Springer}
}

@article{cresswell1999discrete,
  title={The discrete first, second and thirty-fourth {P}ainlev{\'e} hierarchies},
  author={Cresswell, C. and Joshi, N.},
  journal={Journal of Physics A: Mathematical and General},
  volume={32},
  number={4},
  pages={\href{https://iopscience.iop.org/article/10.1088/0305-4470/32/4/009/}{655}},
  year={1999},
  publisher={IOP Publishing}
}

@article{casati2020recursion,
  title={Recursion and {H}amiltonian operators for integrable nonabelian difference equations},
  author={Casati, M. and Wang, J. P.},
  journal={Nonlinearity},
  volume={34},
  number={1},
  pages={\href{https://doi.org/10.1088/1361-6544/aba88c}{205}},
  year={2020},
  publisher={IOP Publishing},
  notes={\href{https://arxiv.org/abs/1910.06807v3}{arXiv:1910.06807}}
}

@article{svinin2009some,
  title={On some class of reductions for the {I}toh--{N}arita--{B}ogoyavlenskii lattice},
  author={Svinin, A. K.},
  journal={Journal of Physics A: Mathematical and Theoretical},
  volume={42},
  number={45},
  pages={\href{https://iopscience.iop.org/article/10.1088/1751-8113/42/45/454021}{454021}},
  year={2009},
  publisher={IOP Publishing},
  note={\href{https://arxiv.org/abs/0902.4517v3}{arXiv:0902.4517}}
}

@article{svinin2011some,
  title={On some class of homogeneous polynomials and explicit form of integrable hierarchies of differential--difference equations},
  author={Svinin, A. K.},
  journal={Journal of Physics A: Mathematical and Theoretical},
  volume={44},
  number={16},
  pages={\href{https://iopscience.iop.org/article/10.1088/1751-8113/44/16/165206}{165206}},
  year={2011},
  publisher={IOP Publishing},
  note={\href{https://arxiv.org/abs/1101.3808v3}{arXiv:1101.3808}}
}

@article{svinin2025volterra,
  title={Volterra map and related recurrences},
  author={Svinin, A. K.},
  journal={arXiv preprint \href{https://arxiv.org/abs/2502.06908}{arXiv:2502.06908}},
  year={2025}
}

@article{svinin2014some,
  title={On some classes of discrete polynomials and ordinary difference equations},
  author={Svinin, A. K.},
  journal={Journal of Physics A: Mathematical and Theoretical},
  volume={47},
  number={15},
  pages={\href{https://iopscience.iop.org/article/10.1088/1751-8113/47/15/155201}{155201}},
  year={2014},
  publisher={IOP Publishing},
  note={\href{https://arxiv.org/abs/1308.5018}{arXiv:1308.5018}}
}

@article{carpentier2022quantisations,
  title={Quantisations of the {V}olterra hierarchy},
  author={Carpentier, S. and Mikhailov, A. V. and Wang, J. P.},
  journal={Letters in Mathematical Physics},
  volume={112},
  number={5},
  pages={\href{https://link.springer.com/article/10.1007/s11005-022-01588-1}{94}},
  year={2022},
  publisher={Springer},
  note={\href{https://arxiv.org/abs/2204.03095}{arXiv:2204.03095}}
}

@article{gordoa2005non,
  title={Non-isospectral lattice hierarchies in 2+1 dimensions and generalized discrete {P}ainlev{\'e} hierarchies},
  author={Gordoa, P. R. and Pickering, A. and Zhu, Z. N.},
  journal={Journal of Nonlinear Mathematical Physics},
  volume={12},
  number={sup2},
  pages={\href{https://doi.org/10.2991/jnmp.2005.12.s2.13}{180--196}},
  year={2005},
  publisher={Taylor \& Francis}
}

@article{svinin2015integrals,
  title={On integrals for some class of ordinary difference equations admitting a {L}ax pair representation},
  author={Svinin, A. K.},
  journal={arXiv preprint \href{https://arxiv.org/abs/1505.06394}{arXiv:1505.06394}},
  year={2015}
}
\end{document}